\documentclass[11pt,reqno]{article}
\usepackage[english]{babel} % Language for the sections
\usepackage[latin1]{inputenc} % Accents
\usepackage{lmodern}
\usepackage{blindtext} % for generating dummy text
\usepackage{abbrevs} % For abbreviations. See section below.

\usepackage{amsmath}
\usepackage{amsthm} % For the theorems
\usepackage{amssymb}
\usepackage{amsfonts} % For the real symbol
\usepackage{mathrsfs} %For the sigma-algebra letter
\usepackage{mathtools} % General mathematics customization
\usepackage{thmtools} % Theorem customization
\usepackage{thmtools, thm-restate} % Theorem customization
\usepackage{dsfont} % To have mathhbb-like format in numbers
\usepackage{scalerel} % To modify the font size of subscripts

\usepackage{graphicx} % To insert graphs
\usepackage{subcaption} % To set several graphs in one page
\usepackage{caption} % To set caption width in a figure
\usepackage{placeins} % The graph is here and nowhere else
\usepackage{tikz} % To draw pictures like sets and decision trees
\usepackage{float} % To force graphs in a specific place
\usepackage{rotating}

\usepackage{booktabs} % Enhances the quality of tables
\usepackage{multirow} % Columns spanning multiple rows
\usepackage{anysize} % Margins
\usepackage[flushleft]{threeparttable} % Comments at the end of the table
\usepackage{longtable} % For tables that go on for more than one page.
\usepackage{threeparttablex} % It's like threeparttable but combined with longtable

\usepackage{color} % To add a bit of color to your document
\usepackage{titlesec} % Titles in sections
\usepackage{natbib} %Citations
\usepackage{bm} % bold italics in mathmode
\usepackage{geometry} % Margins and page layout
\usepackage{comment} % To comment sections
\usepackage{hyperref} % For hyperlinks in the PDF
\usepackage{epigraph} % Inspirational quote at the beginning of chapter
\usepackage{etoolbox}
\usepackage{footmisc} % Changing thanks symbol to dagger
\usepackage{titling} % For changing the color of the thanks symbol
\usepackage{enumerate} % To customize the appearance of the enumerator
\usepackage{csquotes} % To insert quotes
\usepackage{pifont} % Checkmark symbol
\usepackage{setspace} % Different ways to change line spacing
\usepackage{changepage}   % for the adjustwidth environment needed in the Monte Carlo comparisons
\usepackage{enumitem} % To enumerate with alphanumeric
\usepackage{listings} % To customize listings
\usepackage{enumerate}
\usepackage{ifthen} % Conditional Compilation
\usepackage{todonotes} % For \todo{...} notes

\usepackage{pdflscape}

\usepackage{xr} % To cross-reference to the online appendix

\makeatletter
\@addtoreset{section}{part}
\makeatother
\titleformat{\part}[display]
{\normalfont\LARGE\bfseries}{}{0pt}{}

\hypersetup{   
      colorlinks = true,   
      urlcolor = blue,  
      colorlinks = true,
      linkcolor = magenta,
      citecolor = blue,    
      pdfkeywords = {econometrics},
      pdfauthor={Mauricio Olivares, Tomasz Olma, Daniel Wilhelm},    
      pdfpagemode = UseNone 
      }

\newboolean{arxiv}
\setboolean{arxiv}{true} % <- true for arXiv

\declaretheorem{theorem}
\newtheorem{lemma}{Lemma}

\newtheorem{assumption}{Assumption}

\makeatletter
\newcommand{\neutralize}[1]{\expandafter\let\csname c@#1\endcsname\count@}
\makeatother

\declaretheoremstyle[
  spaceabove = 8pt,
  spacebelow = 8pt,
  headfont=\color{black}\bfseries,
  bodyfont=\normalfont,
  qed=$\blacksquare$,
  ]{remark_style}

\declaretheorem[style=remark_style]{remark}

\newtheoremstyle{break}
  {\topsep}{\topsep}%
  {\itshape}{}%
  {\bfseries}{}%
  {\newline}{}%
\theoremstyle{break}
\newtheorem{algorithm}{Algorithm}

\DeclareMathOperator*{\E}{\mathbb{E}}
\newcommand{\Var}{{\rm Var}}

\DeclareMathOperator*{\argmin}{arg\,min}

\renewcommand{\Pr}{{\mathrm{P}}}

\DeclarePairedDelimiter\abs{\lvert}{\rvert}%
\DeclarePairedDelimiter\norm{\lVert}{\rVert}%

\makeatletter
\let\oldabs\abs
\def\abs{\@ifstar{\oldabs}{\oldabs*}}
\let\oldnorm\norm
\def\norm{\@ifstar{\oldnorm}{\oldnorm*}}
\makeatother

\newcommand{\sss}{\scriptscriptstyle}

\mathcode`e=`e 
\newcommand{\ind}{\mathds{1}} 

\newabbrev\cdf{\scalebox{0.9}{CDF}}

\makeatletter
\renewcommand\maybe@space@{%
  \maybe@ictrue % <= this is new
  \expandafter   \@tfor
    \expandafter \reserved@a
    \expandafter :%
    \expandafter =%
                 \nospacelist
                 \do \t@st@ic
  \ifmaybe@ic % <= this is new
    \space
  \fi
}
\makeatother

\usepackage{verbatim,color}
\def\boxit#1{\vbox{\hrule\hbox{\vrule\kern6pt
            \vbox{\kern6pt#1\kern6pt}\kern6pt\vrule}\hrule}}

\DefineFNsymbols{mySymbols}{{\ensuremath\dagger}{\ensuremath\ddagger}\S\P
   *{**}{\ensuremath{\dagger\dagger}}{\ensuremath{\ddagger\ddagger}}}
\setfnsymbol{mySymbols}
\thanksmarkseries{fnsymbol} %To change the color of the thanks mark

\title{\bf \Large A Powerful Bootstrap Test of Independence in High Dimensions\thanks{We thank Fang Han and Thomas Nagler for helpful comments. The authors gratefully acknowledge financial support from the European Research Council (Starting Grant No. 852332).}}

\author{
  \normalsize Mauricio Olivares\\ 
  \normalsize Department of Statistics\\
  \normalsize LMU Munich\\
  \normalsize \url{m.olivares@lmu.de}
  \and
  \normalsize Tomasz Olma\\ 
  \normalsize Department of Statistics\\
  \normalsize LMU Munich\\
  \normalsize \url{t.olma@lmu.de}
  \and
  \normalsize Daniel Wilhelm\\ 
  \normalsize Departments of Statistics and\\
  \normalsize Economics\\
  \normalsize LMU Munich\\
  \normalsize \url{d.wilhelm@lmu.de}
}

\begin{document}

\date{\today} % To omit the date when using \maketitle
\maketitle
\thispagestyle{empty}
% \vspace{-8mm}

\begin{abstract} 
    This paper proposes a nonparametric test of pairwise independence of one random variable from a large pool of other random variables. The test statistic is the maximum of several Chatterjee's rank correlations and critical values are computed via a block multiplier bootstrap. We show in simulations that other popular tests based on distance covariances do not necessarily control size under this null. Our test, on the other hand, is shown to asymptotically control size uniformly over a large class of data-generating processes, even when the number of variables is much larger than sample size. The test is consistent against any fixed alternative. It can be combined with a stepwise procedure for selecting those variables from the pool that violate independence, while controlling the family-wise error rate. All formal results leave the dependence among variables in the pool completely unrestricted. In simulations, we find that our test is typically more powerful than competing methods (in settings where they are valid), particularly in high-dimensional scenarios or when there is dependence among variables in the pool.
\end{abstract}
\vspace{5mm}

% \noindent {\bf JEL Classification:} %C12, C14.

\newgeometry{tmargin=2cm,lmargin=2cm,rmargin=2cm}

% --------------------------------------------------------- %

%                      INTRODUCTION                         %

% --------------------------------------------------------- %

% \doublespacing
\section{Introduction}
\label{sec:introduction}

This paper is concerned with nonparametric testing of pairwise independence between a random variable $X$ and many other random variables $Y_1,\ldots,Y_p$,
$$H_0\colon Y_j\perp X\:\text{ for all }\:j\in\{1,\dots,p\}, $$
against the alternative $H_1$, which is the negation of $H_0$. The goal is to propose a powerful test of $H_0$ allowing for $p$ to be much larger than the sample size while at the same time not restricting the dependence among $Y_1,\ldots,Y_p$ in any way. In a second step, we want to combine the new test with a stepwise procedure for screening out variables from $Y_1,\ldots,Y_p$ that violate independence so as to control the family-wise error rate.

There are many applied examples in which testing $H_0$ and, in particular, screening out variables that violate independence is of interest. For instance, in causal inference, one might want to test whether a treatment indicator has an effect on various outcomes and then select those outcomes on which there is an effect. Such a test could also be applied to ``placebo'' outcomes, i.e. pre-treatment outcomes that the researcher knows cannot have been affected by the treatment, to validate unconfoundedness assumptions. Another example concerns testing fairness of machine learning predictions, where one might want to test independence of a prediction from a set of protected characteristics. As a final example, one might want to test independence of a measure of environmental exposure (e.g., whether or not a person smokes) from a vector of genetic markers. In our empirical application, we use the proposed test for identifying genes whose transcript levels oscillate during the cell cycle.

As test statistic we consider the maximum of $p$ rank correlation coefficients by \cite{chatterjee2021new} for testing independence between $Y_j$ and $X$, $j=1,\ldots,p$. Critical values for the test are computed via a block multiplier bootstrap that perturbs an asymptotically linear representation of the rank correlations. We then show that, as the sample size grows, the proposed test controls size uniformly over a large class of data-generating processes. This result allows for the dimension $p$ to grow exponentially in sample size and the dependence among $Y_1,\ldots,Y_p$ is left completely unrestricted.

While it has been shown that the nonparametric bootstrap does not consistently estimate the distribution of a single Chatterjee's rank correlation \citep{Lin:2024ui}, the proposed block multiplier bootstrap achieves this by accounting for the dependence of the summands in the asymptotically linear representation. The existence of an asymptotic linear representation for Chatterjee's rank correlation together with suitably strong control of the remainder terms allows us to employ recent results for the approximation of maxima of sums of high-dimensional vectors \citep{chernozhukov2013gaussian,chernozhukov2015comparison,chernozhukov2019inference} to establish the validity of the block multiplier bootstrap also for the maximum of many Chatterjee's rank correlation coefficients.

The block multiplier bootstrap requires a tuning parameter choice, namely the block size that should tend to infinity with sample size. We are able to provide a simple choice that enjoys a certain optimality property: it minimises the distance of the bootstrap variance and the asymptotic variance for each individual test statistic. In simulations, however, we find that the size and power of our test are almost insensitive to the particular value of the block size. 

The test is consistent against all fixed alternatives under which $Y_1,\ldots,Y_p,X$ are continuously distributed. Finally, our test is computationally attractive even for very large sample sizes and high dimensions $p$.

The test can be combined with a stepwise multiple testing procedure \citep[as in][]{romano2005exact} for screening out variables from $Y_1,\ldots,Y_p$ that violate independence so as to asymptotically control the family-wise error rate uniformly over a large set of data-generating processes. This result allows for the dimension $p$ to grow exponentially in sample size and the dependence among $Y_1,\ldots,Y_p$ is left completely unrestricted. The companion R package \href{https://github.com/mauolivares/hdIndep}{\texttt{hdIndep}} facilitates the implementation of the methodology developed in this paper.

\paragraph{Related Literature}

Our paper is most closely related to the literature on testing independence of two random vectors, $H_0^{joint}\colon Y \perp X$, where $Y\in\mathbb{R}^p$ and $X\in\mathbb{R}^q$. This is a different hypothesis from the one we consider; for $q=1$, it implies, but is not implied by, $H_0$; when $H_0$ holds, but the copula of $Y_1,\ldots,Y_p$ given $X$ is not independent of $X$, then $H_0^{joint}$ is violated. Therefore, tests that control the rejection probability under $H_0^{joint}$ are not guaranteed to control it under our hypothesis $H_0$. \cite{Sinha:1977re}, \cite{Taskinen:2005ty}, \cite{Bakirov:2006ui}, \cite{Szekely:2007io}, \cite{Heller:2012oi}, \cite{Heller:2012op}, \cite{Shi:2022oi} propose nonparametric tests of $H_0^{joint}$, where $p$ and $q$ are of arbitrary, but fixed (with sample size) dimensions. \cite{Szekely:2013re} show that the test statistic of \cite{Szekely:2007io} is biased in high dimensions and an independence test based on it therefore does not control size in high dimensions. They also propose a bias-corrected test statistic and derive its asymptotic distribution under the null of independence when the dimensions of both vectors grow with the sample size. The asymptotic regime under which their test is valid requires the dimensions of both vectors to grow, so it is not clear (at least to us) whether it is also valid when one of the two dimensions remains constant as the sample size grows. In addition, their derivation of the test statistic's limiting distribution requires the elements of the two vectors to be exchangeable, a condition we do not require for $Y_1,\ldots,Y_p$. \cite{Ramdas:2015aa} show that both independence tests, \cite{Szekely:2007io} and \cite{Szekely:2013re} have low power against ``fair alternatives'' in high dimensions. Our simulations in Section~\ref{sec:invalidity_of_distance_covariance_tests} show that these tests do not necessarily control size under $H_0$.

More recently, \cite{zhou2024rank} and \cite{Wang:2024ui} propose other rank-based tests, e.g. based on Hoeffding's D, Blum-Kiefer-Rosenblatt's R and Bergsma-Dassios-Yanagimoto's $\tau$ among others, of $H_0^{joint}$ in high dimensions. The validity of these tests relies on at least one of the dimensions $p$ and/or $q$ diverging so that a central limit theorem across the elements of, say, $Y$ can be invoked. This approach necessarily restricts the dependence of $Y_1,\ldots,Y_p$, while our validity results leave the dependence completely unrestricted.

In simulations, we find that our test is typically more powerful than the alternatives by \cite{Szekely:2007io}, \cite{Szekely:2013re}, \cite{zhu2020distance}, and \citet{zhou2024rank} (in scenarios in which they are valid) in high dimensions (large $p$) or when there is dependence among the variables $Y_1,\ldots,Y_p$.

Our proposed test has two additional advantages over competitors (in scenarios in which they are valid): (i) unlike tests by \cite{zhu2020distance} and \citet{zhou2024rank} ours is computationally inexpensive and (ii) we develop a stepwise procedure for selecting hypotheses that are not rejected so as to control the family-wise error rate. 

There is a large literature on nonparametric tests of mutual independence among the elements of a random vector. Some examples are \cite{Leung:2018iu}, \cite{Shun:2018ui}, \cite{Drton:2020oi}, \cite{Wang:2024op}, \cite{Bastian:2024io}; see also references therein. \cite{Xia:2024aa} propose such a test based on Chatterjee's rank correlation. While the hypothesis considered in our paper also involves many nonparametric independence tests, it substantially differs from the null of mutual independence. This is because, in our testing problem, $X$ occurs in every independence hypothesis and the dependence of $Y_1,\ldots,Y_p$ is left unrestricted. 

Finally, since our proposed test is based on the rank correlation for two random variables proposed by \cite{chatterjee2021new}, our paper is also related to a recent and fast-growing literature that examines the rank correlation coefficient's properties. \cite{chatterjee2021new} shows asymptotic normality of the correlation coefficient under independence of the two random variables. \cite{Lin:2022aa} and \citet{kroll2024asymptotic} show that it is also asymptotically normal under dependence. \cite{Shi:2021oi} and \cite{Lin:2022io} examine and propose improvements of the power of tests of independence based on the rank correlation coefficient. Based on the earlier work by \cite{chatterjee2021new}, \cite{Azadkia:2021oi} introduce a graph-based correlation coefficient that can be seen as a multivariate extension of Chatterjee's correlation coefficient. \cite{Han:2024io} derive its asymptotic distribution under the null that a vector $Y$ (with fixed dimension) is independent of a random variable $X$, and they find, in simulations, that the test may be more powerful than competitors in higher dimensions. For a recent review of this literature, see \cite{Chatterjee:2024op}.

\section{The Test}\label{sec:test}

This section first introduces the new test, establishes asymptotic size control uniformly over a large class of data-generating processes, and then consistency against all fixed alternatives under which $Y_1,\ldots,Y_p$ and $X$ have continuous distributions. The section concludes with the development of an optimal tuning parameter choice.

Let $\mathbb{D}\coloneqq \{(X_i,Y_{1,i},\ldots,Y_{p,i})\}_{i=1}^n$ be an i.i.d. sample drawn from the distribution of $(X,Y_1,\ldots,Y_p)$. For each individual hypothesis $H_{0,j}\colon Y_j \perp X$ there are many available tests in the literature. In this paper, we focus on the test statistic by \cite{chatterjee2021new}. The motivation for this choice will become clear later in this section. To define the test statistic let $X_{(k)}$ be the $k$-th order statistic of $X_1,\ldots,X_n$, i.e. $X_{(1)} \leq \dots \leq X_{(n)}$, and $Y_{j,(k)}$ be the concomitant of $X_{(k)}$, i.e. if $X_{(k)}=X_l$, then $Y_{j,(k)}=Y_{j,l}$. Denote by $F_{Y_j}$ the cumulative distribution function (cdf) of $Y_j$ and by  $\hat{F}_{Y_j}$ the  empirical cdf. Then, Chatterjee's rank correlation for testing an individual hypothesis $H_{0,j}$ is
\begin{equation}\label{eq: C3}
    \hat\xi_j \coloneqq  1 - \frac{ 3n }{n^2-1}\sum_{i=1}^{n-1} \abs{ \hat{F}_{Y_j}(Y_{j,(i+1)})-\hat{F}_{Y_j}(Y_{j,(i)}) }~.
\end{equation}
\cite{chatterjee2021new} shows that $\hat\xi_j$ is a consistent estimator of 
$$\xi_j \coloneqq \frac{\int Var(\E[\ind\{Y_j\geq t\}|X])f_{Y_j}(t)dt}{\int Var(\ind\{Y_j\geq t\})f_{Y_j}(t)dt}, $$
a measure of dependence introduced by \cite{Dette:2013yy} in the case in which $Y_j$ has a continuous distribution with density $f_{Y_j}$. This measure has several attractive features \citep{chatterjee2021new}. Two features that are particularly important for the test proposed in this paper are that (i) $\xi_j$ is equal to zero if, and only if, $X$ and $Y_j$ are independent, and (ii) the estimator $\hat\xi_j$ admits an asymptotic linear representation (shown in \eqref{eq: angus_representation} below) with a remainder that we can show to be sufficiently small. Our arguments for validity of the proposed test in high dimensions crucially depend on property~(ii).

The proposed test statistic for $H_0$ is the maximum of the individual Chatterjee's rank correlations:
\begin{equation}\label{eq:test_statistic}
    \hat{T} \coloneqq \sqrt{n}\max_{1\le j\le p}\,\hat\xi_j.
\end{equation}
We propose to compute critical values for the test statistic via a block multiplier bootstrap. To describe the procedure consider first an individual hypothesis $H_{0,j}$. If the null $H_{0,j}$ holds and both random variables are continuously distributed, then the arguments in \cite{angus1995coupling} imply that Chatterjee's rank correlation has an asymptotically linear representation of the form
\begin{equation}\label{eq: angus_representation}
    \sqrt{n}\,\hat{\xi}_j =\frac{1}{\sqrt{n}}\sum_{i=1}^{n-1}W_{j,i}+r_{j,n}~,
\end{equation}
where $r_{j,n}$ is a negligible remainder term and $W_{j,i}$ is defined as
$$W_{j,i}\coloneqq 2 - 3\abs{U_{j,i+1}-U_{j,i}} - 6 U_{j,i}(1-U_{j,i}) $$
with $U_{j,i} \coloneqq F_{Y_j}(Y_{j,(i)})$. A naive application of the multiplier bootstrap idea would be to repeatedly draw bootstrap multipliers $\varepsilon_1,\ldots,\varepsilon_n$ as independent standard normal random variables that are independent of the data $\mathbb{D}$ and then compute a critical value from the distribution of $\frac{1}{\sqrt{n}}\sum_{i=1}^{n-1}\varepsilon_i W_{j,i}$ given the data. However, there are two problems with this approach. First, $\{W_{j,i}\}_{i=1}^n$ is not an i.i.d. sequence but a 1-dependent process. Second, $\{W_{j,i}\}_{i=1}^n$ is not directly observed because $F_{Y_j}$ is unknown. 

To address the first challenge, we decompose the sum $\sum_{i=1}^{n-1} W_{j,i}$ into the sums over ``big'' and ``small'' blocks formed of $\{W_{j,i}\}_{i=1}^n$ with the property that the big blocks are independent of each other. Formally, let $q\ge 1$ denote the size of big blocks. It is a tuning parameter to be chosen by the researcher; in Section~\ref{sec:choice of q} we develop an optimal choice of $q$. Since $\{W_{j,i}\}_{i=1}^n$ is a 1-dependent sequence, we let the small blocks be of length one. 
Further, let $m\coloneqq \lfloor(n-1)/(q+1)\rfloor$, where $\lfloor\nu\rfloor$ is the integer part of $\nu$, denote the number of big blocks. Lastly, $r \coloneqq n-1 - m(q+1)$,  $0\le r<q+1$, is the number of remaining summands that are not allocated to any of the small or big blocks. With this notation, we can thus write 
$$\sum_{i=1}^{n-1}W_{j,i} = \sum_{k=1}^m A_{j,k} + \sum_{k=1}^m B_{j,k} + R_j,$$
where
\[
	A_{j,k} \coloneqq \sum^{kq+(k-1)}_{l=(k-1)(q+1)+1} W_{j,l}, \quad B_{j,k} \coloneqq W_{j,k(q+1)}, \quad \text{ and } \quad R_j \coloneqq \sum_{k=1}^r W_{j,m(q+1)+k}.
\]
The big block sums $A_{j,1},\ldots,A_{j,m}$ are independent of each other, and we show that the terms $B_{j,1},\ldots,B_{j,m}$ and $R_j$ are asymptotically negligible. It then follows that $\sqrt{n}\,\hat{\xi}_j $ can be approximated by $\frac{1}{\sqrt{mq}}\sum_{k=1}^m A_{j,k}$, the sum of independent components.

We now need to address the second challenge, which is that $W_{j,i}$, and thus also $A_{j,k}$, are not observed. $W_{j,i}$ can be estimated by
\[
\hat W_{j,i}\coloneqq 2 - 3\abs{ \hat U_{j,i+1}- \hat U_{j,i}} - 6 \hat U_{j,i}(1-\hat U_{j,i}),
\]
where $\hat U_{j,i}\coloneqq \hat{F}_{Y_j}(Y_{j,(i)})$, and $A_{j,k}$ by
\[
\hat{A}_{j,k} \coloneqq  \sum^{kq+(k-1)}_{l=(k-1)(q+1)+1} \hat{W}_{j,l}.
\]
While $A_{j,1},\ldots,A_{j,m}$ are independent, $\hat{A}_{j,1},\ldots,\hat{A}_{j,m}$ are only asymptotically independent, i.e. in the limit as $n,m\to\infty$.

Finally, bootstrap multipliers $\varepsilon_1,\ldots,\varepsilon_m$ are drawn as independent standard normal random variables that are independent of the data $\mathbb{D}$. The bootstrap statistic is then defined as
\begin{equation}\label{eq:bootstrap_statistic}
    \hat{T}^B\coloneqq \max_{1\leq j \leq p} \frac{1}{\sqrt{mq}} \sum_{k=1}^m \varepsilon_k \hat{A}_{j,k}.
\end{equation}
For a nominal level $\alpha\in(0,1)$, the proposed critical value $\hat c(\alpha)$ for our test is the conditional $(1-\alpha)$-quantile of $\hat{T}^B$ given the data $\mathbb{D}$. The test rejects $H_0$ if, and only if, the test statistic $\hat T$ exceeds the critical value $\hat c(\alpha)$.

\subsection{Size Control}
\label{sec:size control}
In this subsection, we show that our proposed test asymptotically controls size uniformly over a large class of data-generating processes.

\begin{assumption}\label{ass:continuity}
    $X, Y_1,\ldots,Y_p$ are continuously distributed.
\end{assumption}

Assuming all random variables have a continuous distribution simplifies the presentation, but is not essential. Remark~\ref{rem: discreteness} below discusses extensions to the case with discrete distributions.

\begin{assumption}\label{ass:rates}
    Suppose $p\geq 2$. There exist constants $C_1>0$ and $0<\gamma<1/4$ such that
    $(1/q)\log^2 p \leq C_1 n^{-\gamma}$ and $\max\big\{q\log^{5/2}p,\,\sqrt{q}\log^{7/2}(pn) \big\} \leq C_1 n^{1/2-\gamma}$.
\end{assumption}

This assumption requires $q$ to diverge as the sample size grows, but restricts its rate to be neither too slow nor too fast. The assumption also restricts the rate at which the dimension $p$ is allowed to grow with sample size. However, the upper bound on the rate is very large: $p$ may be an exponential function of sample size and thus is allowed to be much larger than sample size. For instance, there are positive constants $\delta_1,\delta_2$ so that $p = e^{n^{\delta_1}}$ and $q = n^{\delta_2}$ satisfy Assumption~\ref{ass:rates}.

\begin{theorem}\label{thm: size control}
    Suppose that Assumptions~\ref{ass:continuity}--\ref{ass:rates} hold. Then, under the null hypothesis $H_0$, there exist positive constants $c$, $C$ depending only on $\gamma$ and $C_1$ such that  
    \begin{equation}\label{eq: size control}
        \left|\Pr\left( \hat{T} > \hat{c}(\alpha)\right) - \alpha\right| \leq C n^{-c}.
    \end{equation}
\end{theorem}

The result in \eqref{eq: size control} implies that the proposed test asymptotically controls size. In fact, the asymptotic size is equal to the nominal level $\alpha$ and, in this sense, the test is not conservative. Furthermore, the probability of rejecting $H_0$ when $H_0$ is satisfied can deviate from the nominal level only by a term that is polynomially small in $n$. Importantly, the constants $c$ and $C$ depend on the data-generating process only through the constants $\gamma$ and $C_1$ from Assumption~\ref{ass:rates}. Therefore, under $H_0$, the inequality in \eqref{eq: size control} holds uniformly over all data-generating processes that satisfy the assumption with the same constants, denoted by $\mathbf{P}_{\gamma,C_1}$:
$$\limsup_{n\to\infty}\sup_{P\in\mathbf{P}_{\gamma,C_1}} \left|\Pr\left( \hat{T} > \hat{c}(\alpha)\right) - \alpha\right| =0, $$
i.e. our test has asymptotic size equal to $\alpha$ uniformly over those data-generating processes.

It is worth noting that the validity of our test is guaranteed in high dimensions without restricting the dependence of $Y_1,\ldots,Y_p$ in any way.

To establish this result we need to show that the distribution of the bootstrap statistic $\hat{T}^B$ given the data is close to the distribution of the test statistic $\hat{T}$. This is achieved in several steps: (i) show that $\hat{T}$ is close to 
$$T_0\coloneqq \max_{1\le j\le p}\,\frac{1}{\sqrt{n}}\sum_{i=1}^{n-1}W_{j,i}~,$$
(ii) show that $\hat{T}^B$ is close to 
$$T_0^{B}\coloneqq\max_{1\leq j \leq p} \frac{1}{\sqrt{mq}}\sum_{k=1}^m \varepsilon_k A_{j,k},$$
and then (iii) show that both $T_0$ and $T_0^B$ are close to the maximum of Gaussian random variables,
$$Z_0\coloneqq \max_{1\le j\le p} V_j,$$
where $V\coloneqq (V_1,\dots,V_p)'$ is a Gaussian vector such that $\E(V)=0$ and $\E(V\,V')=\frac{1}{mq}\sum_{k=1}^m\E(A_k A_k')$, $A_k\coloneqq (A_{1,k},\dots,A_{\,p,k})'$. Steps (i) and (ii) are developed in the proof of the theorem and step (iii) is delegated to Lemmas~\ref{lem: CCKB1} and \ref{lem: CCKB2}. These three steps establish that the distribution of the test statistic is close to that of the bootstrap statistic. However, this result does not yet imply the statement in \eqref{eq: size control} because $\hat{c}(\alpha)$ is random and may be correlated with $\hat{T}$. The final step of the proof therefore consists of passing from a deterministic to a random critical value.

Step (iii) of the derivation combines several results on high-dimensional Gaussian approximations from \cite{chernozhukov2013gaussian,chernozhukov2015comparison,chernozhukov2019inference}. The main challenge of the proof is contained in steps (i) and (ii), where we need to establish that all remainder terms $r_{1,n},\ldots,r_{p,n}$ from the representation \eqref{eq: angus_representation} vanish at a suitably fast rate. \cite{angus1995coupling} shows that, for each $j$, $r_{j,n}=o_P(n^{-1/2})$, but in our high-dimensional setting, we need stronger control of these remainder terms, and these results are developed in Lemma~\ref{lem: verify_Assumption}.

\begin{remark}[discrete distributions]\label{rem: discreteness}
    If any of $X,Y_1,\ldots,Y_p$ have a noncontinuous distribution, then they can be replaced by new random variables that are equal to the original ones plus a sufficiently small amount of noise (which is independent of the data). Our proposed test can then be applied to the new random variables. One can show that the test continues to be valid in this case because the null hypothesis $H_0$ (in terms of the original variables) implies that the probability limits of Chatterjee's rank correlations in terms of the new variables are all equal to zero.
\end{remark}

\begin{remark}[bootstrapping Chatterjee's rank correlation]
    \cite{Lin:2024ui} show that, under the null of independence of two random variables, the nonparametric bootstrap does not consistently estimate the limiting distribution of Chatterjee's rank correlation. In particular, they show that the bootstrap yields a variance estimate whose expectation is below $2/5$, the correct limiting variance of Chatterjee's rank correlation under the null of independence and Assumption~\ref{ass:continuity}.

    Assumption~\ref{ass:rates} requires $p\geq 2$, but an inspection of the proof of Theorem~\ref{thm: size control} reveals that the result can also be proven for $p=1$ under slightly simplified rate conditions. Therefore, the block multiplier bootstrap correctly approximates the limiting distribution of Chatterjee's rank correlation under the null of independence. The reason for this is that it correctly accounts for the 1-dependence in the asymptotic linear representation, and thus correctly estimates the variance of Chatterjee's rank correlation, while the nonparametric bootstrap ignores this dependence.\footnote{\cite{Dette:2024re} show that the m-out-of-n bootstrap is also valid.}
\end{remark}

\begin{remark}[studentisation]\label{rem: studentization}
    The test proposed in this section does not studentise the test statistics. In simulations in Section~\ref{sec:simulations}, we find that studentisation may improve size and power of the test. We consider studentising the individual test statistics by their standard deviation under the null, i.e.,
    $$\hat T^{stud} \coloneqq \sqrt{n} \max_{1\leq j\leq p}  \frac{\hat\xi_j}{\sqrt{v_n}}. $$
    where
    \begin{equation*}
	    v_n \coloneqq \frac{n(n-2)(4n-7)}{10(n-1)^2(n+1)} =  \Var\big( \sqrt{n}\hat\xi_j \big)
    \end{equation*}
    under the null \citep[][Lemma~2]{zhang2023asymptotic}. The sequence $v_n$ is monotonically increasing and, as $n \to \infty$, it converges to $2/5$, the asymptotic variance derived by \citet{chatterjee2021new}.
    
    For the bootstrap statistics there are at least two different possibilities for studentisation. First, one could studentise it by the square root of $\E[A_{j,k}^2]/q = 0.4 + 0.1/q$, i.e.
    \begin{equation*}
        \hat T^{B,stud\,1} \coloneqq \max_{1\leq j \leq p} \frac{1}{\sqrt{m}} \sum_{k=1}^m \varepsilon_k \frac{\hat{A}_{j,k}}{\sqrt{0.4q + 0.1}}.
    \end{equation*}
    This standardisation ensures that the diagonal elements of the bootstrap covariance matrix of individual tests are all approximately centred at one for any $q$, and hence approximately match the variance of $\sqrt{n}\hat\xi_j/\sqrt{v_n}$.
    The second possibility is to employ a bootstrap test statistic in which the big blocks sums are demeaned and standardised by their sample standard deviation, i.e.
    \begin{equation*}
        \hat T^{B,stud\,2} \coloneqq \max_{1\leq j \leq p} \frac{1}{\sqrt{m}} \sum_{k=1}^m \varepsilon_k \frac{\hat{A}_{j,k} - \frac{1}{m} \sum_{k=1}^m \hat{A}_{j,k} }{ \sqrt{\frac{1}{m}\sum_{k=1}^m \hat{A}_{j,k}^2} }.
    \end{equation*}
    This standardisation ensures that the diagonal elements of the bootstrap covariance matrix of individual test statistics are all equal to one for any $q$.
\end{remark}

\subsection{Consistency}
\label{sec:consistency}

The following theorem shows that the proposed test is consistent against any (fixed) violation of the null:
\begin{theorem}\label{thm: consistency_theorem}
    Suppose that Assumptions~\ref{ass:continuity}--\ref{ass:rates} hold. Then, under the alternative hypothesis $H_1$,
    \begin{equation*}
        \Pr\left( \hat{T}> \hat c(\alpha)\right) \to 1\qquad \text{as } n\to\infty.
    \end{equation*}
\end{theorem}
In light of recent discussions related to the local asymptotic power of Chatterjee's test of independence (e.g., \cite{Shi:2021oi} and \cite{Lin:2022io}), it would be interesting to study the power of our proposed test to local alternatives. We view this analysis as beyond the scope of the paper and relegate it to future research.

\subsection{Choice of \texorpdfstring{$q$}{q}}\label{sec:choice of q}
Theorem~\ref{thm: size control} shows that our proposed block multiplier bootstrap test asymptotically controls size under rate conditions on $q$, the size of blocks used to construct the critical value $\hat c(\alpha)$.
These rate conditions specify that $q\to \infty$ as $n \to \infty$ at a rate that is neither too fast nor too slow, but they do not provide any guidance on how to choose $q$ in finite samples. In this section, we develop a choice of $q$ that is optimal in a certain finite-sample sense.

To describe the optimality criterion, recall from the discussion following Theorem~\ref{thm: size control} that the distribution of the bootstrap test statistic $\hat T^B$ is, to first order, determined by its infeasible counterpart $T_0^B=\max_{1\leq j \leq p}T^B_{0,j}$, where
$$T^B_{0,j} \coloneqq \frac{1}{\sqrt{mq}} \sum_{k=1}^m \varepsilon_k A_{j,k} $$
with bootstrap multipliers $\varepsilon_1,\ldots,\varepsilon_m$ that are i.i.d. standard normal and independent of the data. By construction, $T^B_{0,j}$ follows a normal distribution with mean zero and variance that depends on the data and, implicitly, on $q$,
\[
    T^B_{0,j} | \mathbb{D} \sim \mathcal{N}\left(0,V^B_j\right), \quad V^B_j \coloneqq \Var\left( T^B_{0,j}| \mathbb{D} \right).
\]
In our bootstrap procedure, $T^B_{0,j}$ mimics the behaviour of the individual test statistic $\sqrt{n} \hat \xi_j$, which is asymptotically normally distributed and, under the null, has variance $v_n$ given in Remark~\ref{rem: studentization}. Since $q$ affects the conditional distribution of $T^B_{0,j}$ only through the bootstrap variance $V^B_j$ and it does not affect the test statistic itself, we aim to choose $q$ such that $V^B_j$ is close to $v_n$.

The following lemma characterizes the expectation and the variance of the bootstrap variance $V^B_j$ under the null $H_{0,j}$.
\begin{lemma}\label{lemma:VB}
    Suppose that Assumption~\ref{ass:continuity} and the null $H_{0,j}$ hold. For any $j=1,\ldots,p$, it holds that
    \begin{align*}
        \E[V_j^B] & = \frac{2}{5} + \frac{1}{10q},\qquad \text{ for any } q\ge 1, \\
        \Var(V_j^B) & = \frac{1}{m} \begin{cases}
            \frac{7}{20}     & \text{if } q  = 1, \\
            \frac{1353}{2800}    & \text{if } q  = 2, \\
            \frac{8}{25} + \frac{88}{175q} - \frac{229}{700q^2}    & \text{if } q  \ge 3, 
            \end{cases}
    \end{align*}
 	where $m \coloneqq \lfloor(n-1)/(q+1)\rfloor$.
\end{lemma}
Lemma~\ref{lemma:VB} shows that the expectation of $V_j^B$ is above 2/5 for any fixed $q$ and it monotonically converges to 2/5 as $q$ grows large. Since the target variance $v_n$ approaches 2/5 from below (see Remark~\ref{rem: studentization}), this means that $V_j^B$ is biased upwards with a bias that vanishes asymptotically when $q,n \to \infty$. The variance of $V_j^B$, in turn, is generally increasing in $q$. Figure~\ref{fig: illustration} illustrates these relationships for samples of size $n=500$ and $n=1000$. We resolve this bias-variance trade-off by considering the mean squared error of the bootstrap variance $V_j^B$:  
$$MSE_{j,n}(q)\coloneqq \E\left[\left(V_j^B - v_n\right)^2\right] = \left\lfloor\frac{n-1}{q+1}\right\rfloor^{-1}\left.\begin{cases}
            \frac{7}{20}     & \text{if } q  = 1 \\
            \frac{1353}{2800}    & \text{if } q  = 2 \\
            \frac{8}{25} + \frac{88}{175q} - \frac{229}{700q^2}    & \text{if } q  \ge 3 \\
            \end{cases}\right\} + \left(\frac{2}{5} + \frac{1}{10q} - v_n\right)^2. $$
Minimising $MSE_{j,n}(q)$ over $q\in\mathbb{N}^+$ yields our proposed optimal choice of $q$:
\[
	q^*(n) \coloneqq \argmin_{q\in\mathbb{N}^+} MSE_{j,n}(q).
\]
Since $MSE_{j,n}(q)$ is a known function of the sample size, the optimal choice $q^*(n)$ does not depend on the data (beyond $n$). The reason for that is that the individual bootstrap statistic $T^B_{0,j}$ depends on the data only through the ranks of the concomitants, $U_{j,i}\coloneqq F_{Y_j}(Y_{j,(i)})$, and these are independent and uniformly distributed under the null. In consequence, the optimal choice $q^*(n)$ is independent of $j$ and thus the same for each independence hypothesis $H_{0,j}$. 

\begin{figure}[t]
	\centering
	\includegraphics[width=0.95\textwidth]{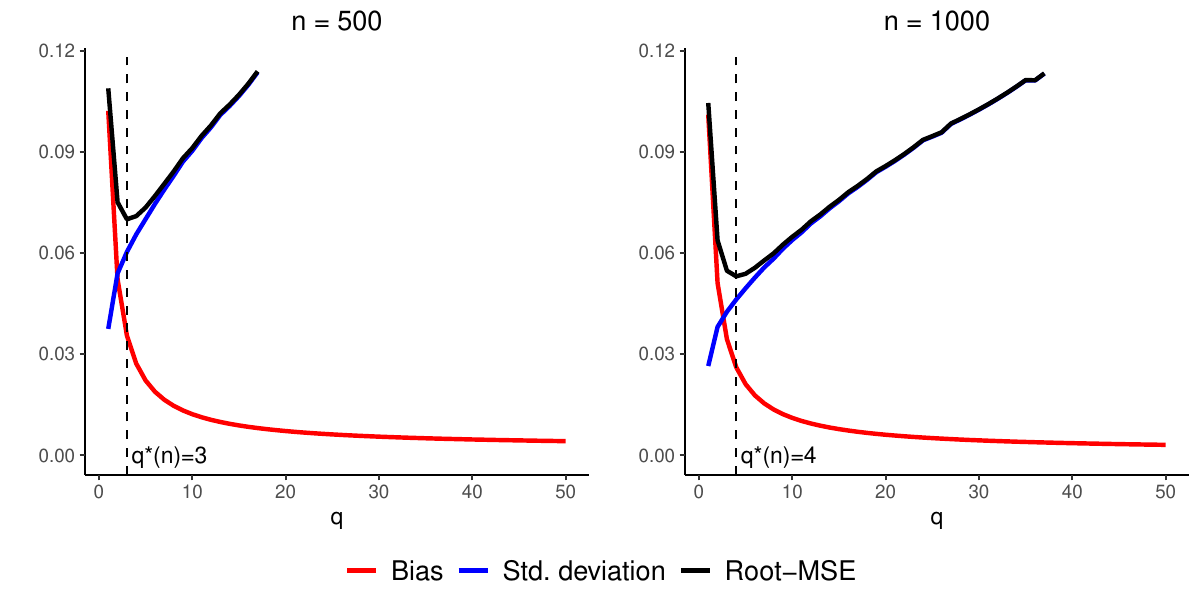}
	\caption{Bias, standard deviation, and the root mean squared error of $V_j^B$ based on the formulas in Lemma~\ref{lemma:VB} together with the optimal choice $q^*(n)$.}\label{fig: illustration}
\end{figure}

The minimiser $q^*(n)$ can be computed for each $n$ by evaluating the function $MSE_{j,n}(q)$ over a grid of values $q \in \{1,2,\ldots, \lfloor (n-1)/2 \rfloor\}$.
On most of its domain, $q^*(n)$ is a step function with large flat regions, but it can oscillate in small transition regions. For example,
%
\begin{comment}
$$q^*(n) = \begin{cases}
	1& \text{if } 3\leq n \leq 87\\
	2& \text{if } 88\leq n \leq 224  \text{ or } n \in I_2 \text{ for some } I_2 \subset \{226,\ldots,244\} \\
	3& \text{if } 245 \leq n \leq 615 \text{ or } n  \in I_3 \text{ for some } I_3 \subset \{225,\ldots,243\} \cup \{617,\ldots,645\} \\
	4& \text{if } 646 \leq n \leq 1344 \text{ or } n  \in I_4 \text{ for some } I_4 \subset \{616,\ldots,644\} \cup \{1346,\ldots,1392\} \\
	\ldots
\end{cases} $$
\end{comment}
%
$$q^*(n) = \begin{cases}
	1 & \text{if } 3\leq n \leq 87, \\
	2 & \text{if } 88\leq n \leq 224, \\
	2 \text{ or } 3 &\text{if } 225 \leq n \leq 244, \\
	3 & \text{if } 245 \leq n \leq 615 , \\
	3 \text{ or } 4 &\text{if } 616 \leq n \leq 645, \\
	4 & \text{if } 646 \leq n \leq 1344, \\
	\ldots
\end{cases} $$
where $q^*(225)=3$, $q^*(n)=2$ for $n \in \{226,\ldots,232\}$, $q^*(233)=3$, etc.
The non-monotone ranges make it impractical to tabularise $q^*(n)$. Instead, we approximate $q^*(n)$ using a smooth function.
To motivate this approximation, note that in a setting where $n$ and $q$ are large but $q$ is much smaller than $n$, which is consistent with the asymptotic regime of Assumption~\ref{ass:rates}, $MSE_{j,n}(q)$ is close to 
$
0.32 q/n + 0.01/q^2$, which is minimised at
$$
	\tilde q(n) \coloneqq (n/16)^{1/3}.
$$
This expression turns out to provide a good approximation to $q^*(n)$ even for small $n$ in the sense that $|\tilde{q}(n) - q^*(n)| < 1$ for any $n \in \mathbb{N}^+$. This property implies that
\[
	q^*(n) = \begin{cases}
		\lceil \tilde q(n)\rceil \text{ if } MSE_{j,n}\big(\lceil \tilde q(n)\rceil \big) \leq MSE_{j,n}\big(\lfloor \tilde q(n)\rfloor\big),  \\
		\lfloor \tilde q(n)\rfloor  \text{ otherwise.}
	\end{cases}
\]
The above approximating property of $\tilde q(n)$ is illustrated in Figure~\ref{fig:optq}.

\begin{remark}[compatibility with rate conditions]
	The optimal choice of $q$ diverges at the rate $n^{1/3}$. This rate is compatible with Assumption~\ref{ass:rates} as long as $p = O(e^{n^a})$ for some $a<1/15$.
\end{remark}

\begin{figure}
	\centering
	\includegraphics[width=0.8\textwidth]{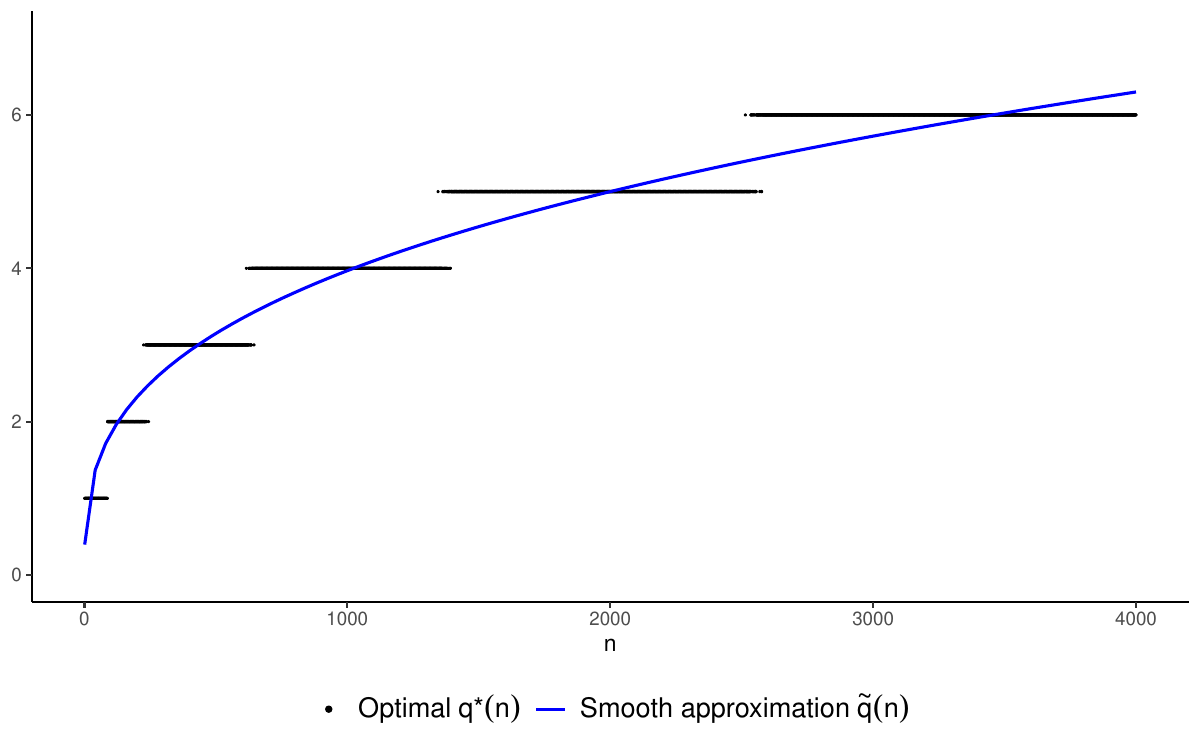}
	\caption{Optimal choice $q^*(n)$ and the smooth approximation $\tilde q(n)  \coloneqq (n/16)^{1/3} $.}
	\label{fig:optq}
\end{figure}

We note that the goal of this paper is to test the hypothesis $H_0$ that all individual hypotheses $H_{0,j}$, $j=1,\ldots,p$, hold simultaneously, while we derived the optimal $q^*(n)$ for an individual test statistic. Therefore, this choice does not necessarily minimise the distance between the distributions of the max-test statistic $\hat{T}\coloneqq \max_{1\leq j\leq p} \sqrt{n}\hat\xi_j$ and the bootstrap statistic $\hat{T}^B \coloneqq \max_{1\leq j \leq p} \hat{T}^B_j$ in any sense. However, minimising the distance between these two distributions is considerably more difficult because the individual statistics may be arbitrarily dependent and the optimal $q$ would then depend on their (unknown) dependence structure. Developing a feasible version of this approach would require estimation of the copula of a high-dimensional random vector, which our proposal above avoids.

In simulations in Section~\ref{sec:simulations}, we find that our proposed choice of $q$ not only optimises the bootstrap approximation of the individual statistics, but also yields a good bootstrap approximation for the max-statistic.

\section{Stepwise Procedure}
\label{sec:stepwise procedure}

In the previous section, we considered testing the null $H_0$ that all variables $Y_1,\ldots,Y_p$ are independent of $X$. Now, consider the problem of selecting individual hypotheses $H_{0,j}\colon Y_j\perp X$ that are violated. The previous section already yields a (``single-step'') method of selection: simply select all hypotheses $H_{0,j}$ for which $\sqrt{n}\hat \xi_j > \hat c(\alpha)$. This section introduces a stepdown procedure that improves upon the single-step procedure by possibly rejecting more hypotheses in finite samples. In addition, we show that the stepdown procedure (and, thus, by extension also the single-step procedure) guarantees asymptotic control of the family-wise error rate.

Let $J(P)\subseteq \{1,\ldots,p\}$ denote the set of hypotheses $H_{0,j}$ that are true under $P$. The family-wise error rate is defined as the probability of rejecting at least one true hypothesis,
$$FWER_P \coloneqq \Pr(\text{reject at least one } H_{0,j}\colon j\in J(P)). $$
For any $I\subseteq \{1,\ldots,p\}$, let
$$\hat T(I) \coloneqq \sqrt{n} \max_{j\in I} \hat \xi_j\qquad\text{and}\qquad \hat T^B(I) \coloneqq \max_{j\in I} \frac{1}{\sqrt{mq}} \sum_{k=1}^m \varepsilon_k \hat A_{j,k}, $$
where $\varepsilon_k$ and $\hat A_{j,k}$ are the multipliers and estimated blocks as introduced in the previous section. Finally, define $\hat c(\alpha;I)$ as the ($1-\alpha$)-quantile of $\hat T(I)$ given the data $\mathbb{D}$.

The following algorithm introduces the stepdown procedure.

\begin{algorithm}[Stepdown Procedure]\label{alg: step-down} Initialise $I_0=\{1,\ldots,p\}$ and $s=0$.
    \begin{enumerate} 
        \item Compute $\hat T(I_s)$ and $\hat c(\alpha;I_s)$.
        \item Is $\hat T(I_s)>\hat c(\alpha;I_s) $ satisfied?
            \begin{enumerate}
                \item {\bf yes:} Reject any hypothesis $H_{0,j}$ with $j\in I_s$ for which $\sqrt{n}\hat \xi_j > c(\alpha;I_s)$, then let $I_{s+1}\subset I_s$ denote the set of hypotheses that have not previously been rejected, set $s \to s+1$, and return to Step 1.
                \item {\bf no:} Stop.
            \end{enumerate}
    \end{enumerate}
\end{algorithm}

\begin{theorem}\label{thm: FWER control}
    Suppose that Assumptions~\ref{ass:continuity}--\ref{ass:rates} hold. Then, the procedure for rejecting individual hypotheses defined in Algorithm~\ref{alg: step-down} satisfies
    $$\limsup_{n\to\infty}\sup_{P\in\mathbf{P}_{\gamma,C_1}} FWER_P \leq \alpha$$
\end{theorem}

The theorem shows that the stepdown procedure in Algorithm~\ref{alg: step-down} asymptotically controls the family-wise error rate uniformly over data-generating processes in $\mathbf{P}_{\gamma,C_1}$.

\section{Simulations}\label{sec:simulations}

In this section, we report results from a series of extensive simulation experiments in which we studied the finite-sample performance of our new test and compared it to existing tests.

\subsection{Overview and Setup}

First, we examine the influence of the choice of block size $q$ on the finite-sample performance of our test. Having established that the test's rejection frequency is fairly insensitive to $q$, in subsequent experiments, we only consider our test with the optimal choice derived in Section~\ref{sec:choice of q}. In the second set of experiments, we extensively study size control and power of our test. We consider a variety of data-generating processes. Key parameters that we vary in the simulations are the degree of dependence among $Y_1,\ldots,Y_p$, the dimension $p$, and whether alternatives are sparse or dense. Having confirmed the theoretical findings that our test controls size and has power against all alternatives, we then show that existing tests based on distance covariance do not control size under our null hypothesis. They are valid only under the stronger null $H_0^{joint}$ mentioned in the introduction, i.e. when the copula of $Y_1,\ldots,Y_p$ given $X$ does not depend on $X$. Finally, we compare our test to existing tests in the special case in which the copula of $Y_1,\ldots,Y_p$ given $X$ does not depend on $X$, because these other tests are valid in that case.

\subsubsection{Data-Generating Processes} 

Let $X$ be distributed uniformly on $[-1,1]$, and let $(\epsilon_1,\ldots,\epsilon_p) \sim \mathcal N(0, \Sigma_\tau)$ be a random vector independent of $X$, where $\Sigma_\tau  \in \mathbb{R}^{p \times p} $ has diagonal elements equal to one and off-diagonal elements equal to $\tau$.
We consider a range of deviations from the null of pairwise independence that differ in the functional form of the association between $X$ and elements of $Y_1, \ldots, Y_p$, as well as in the number and strength of violations of individual hypotheses. In all considered models, the magnitude of the violations is parameterised by $\rho \in \mathbb{R}$, with the null hypothesis corresponding to $\rho=0$, while $\tau$ is the correlation between $Y_1, \ldots, Y_p$ under the null.

\begin{description}
    \item[Model~1~(Sparse linear alternatives).] 
    $Y_1 = \rho X + \varepsilon_1$ and $Y_j = \epsilon_j$ for $j \in \{2,\ldots, p\}$.
    \item[Model~2~(Sparse cosine alternatives).] 
    $Y_1 = \rho \cos(8 \pi X)  + \varepsilon_1$ and $Y_j = \epsilon_j$ for $j \in \{2,\ldots, p\}$.
    \item[Model~3~(Decaying linear alternatives).] 
    $Y_j =  0.9^{\,j-1} \rho X  + \epsilon_j$ for $j \in \{1,\ldots,p\}$.
    \item[Model~4~(Decaying cosine alternatives).] 
    $Y_j = 0.9^{\,j-1} \rho \cos(8 \pi X)   + \epsilon_j$ for $j \in \{1,\ldots,p\}$.
    \item[Model~5~(Dense linear alternatives).] 
    $Y_j = \rho X  + \epsilon_j$ for $j \in \{1,\ldots,p\}$.
    \item[Model~6~(Dense cosine alternatives).] 
    $Y_j = \rho  \cos(8 \pi X)  + \epsilon_j$ for $j \in \{1,\ldots,p\}$.
\end{description}

In Models 1 and 2, $\rho$ determines the magnitude of the violation of the first hypothesis $H_{0,1}\colon Y_1\perp X$, while all other hypotheses, $H_{0,j}\colon Y_j\perp X$ for $j \in \{2,\ldots, p\}$, hold. In Models 5 and 6, all individual hypotheses are violated to the same degree, while Models 3 and 4 offer an intermediate scenario where the magnitude of the violations is strongest for the first hypothesis and decays exponentially with the index of the hypothesis. The individual violations are reparameterised versions of data-generating processes considered in the simulation study by \citet{chatterjee2021new}. Note that under the null, when $\rho = 0$, all six models are identical.

We include cosine alternatives in these experiments partly because, in our empirical application, we conjecture such alternatives to be reasonable departures from independence.

\subsubsection{Tests}

We consider three variants of our proposed test that employ different types of studentisation. Implementations can be found in our accompanying R package \href{https://github.com/mauolivares/hdIndep}{\texttt{hdIndep}}.

\begin{description}
    \item[BMB0.] This test constructs the test statistic $\hat T$ based on the individual Chatterjee's rank correlations $\hat\xi_j$ as in \eqref{eq:test_statistic} and employs the critical value $\hat c(\alpha)$ from the distribution of the bootstrap statistic $\hat T^B$ as in \eqref{eq:bootstrap_statistic}. 
    \item[BMB1.] This test is a variant of BMB0 that uses the studentised test statistic $\hat{T}^{stud}$ and the critical value from the bootstrap statistic $\hat{T}^{B,stud\,1}$ as in Remark~\ref{rem: studentization}.
    \item[BMB2.] This test is a variant of BMB0 that uses the studentised test statistic $\hat{T}^{stud}$ and the critical value from the bootstrap statistic $\hat{T}^{B,stud\,2}$ as in Remark~\ref{rem: studentization}.
\end{description}

For the special case in which the copula of $Y_1,\ldots,Y_p$ given $X$ does not depend on $X$, we also compare our tests to the following alternatives:
\begin{description}
    \item[dcov.test] The distance covariance test of \citet{Szekely:2007io}.
    \item[dcorT.test] The distance correlation t-test for high dimension proposed by \citet{Szekely:2013re}.
    \item[ZXZL] Various rank-based tests proposed in \citet{zhou2024rank}: Hoeffding's D (ZXZL-D), Blum-Kiefer-Rosenblatt's R (ZXZL-R) and Bergsma-Dassios-Yanagimoto's $\tau^*$ (ZXZL-$\tau^*$).
    \item[ZZYS-agg.dcov] The test based on aggregation of marginal distance covariances proposed by \citet{zhu2020distance}.
\end{description}

\subsubsection{Implementation Details}

All simulations are based on $B=499$ bootstrap replications, and $S=5000$ Monte Carlo draws, except in Experiment~4.2, where $S$ is reduced to 1000 due to mitigate computational burden. The significance level is set to $\alpha=0.05$ throughout. Additional results are reported in Appendix~\ref{sec:additional_simulation_results} and are omitted from the main text for brevity.

\subsection{Experiment 1: BMB's Insensitivity to Block Size}

The first set of simulations concerns the rejection rates of our proposed test for different choices of the block size $q$. Figure~\ref{fig:SimI_tau0} presents the results under the null in Panel~A, and under the alternatives specified by Models~1--6 in Panels~B--G. We focus on the case in which there is no correlation between $Y_1,\ldots,Y_p$ under the null ($\tau=0$). The results are very similar when $\tau=0.5$; see Figure~\ref{fig:SimI_tau05} in Appendix~\ref{sec:additional_simulation_results}.

First, we note that our baseline procedure BMB0 controls the size across different scenarios for all considered values of $q$. It is, however, conservative in higher dimensions. The fact that the rejection rate is particularly low for $q=1$ is consistent with the upward bias in the bootstrap variance characterised in Lemma~\ref{lemma:VB}. The optimal rule from Section~\ref{sec:choice of q} yields $q^*(n)=3$ for $n=500$, marked by the vertical dashed lines in the graphs. This choice proves reasonable in all the considered scenarios. 
The issue of conservativeness is alleviated by the studentisation. Both BMB1 and BMB2 effectively shrink the bootstrap distribution and yield rejection probabilities very close to the nominal level of 5\% for small values of $q$. Since the correction in BMB1 is negligible for large $q$, the blue and black lines approach each other as $q$ increases. BMB2 maintains rejection rates closer to 5\% as $q$ grows, but it generally slightly overrejects. 

Panels~B--G of Figure~\ref{fig:SimI_tau0} indicate that all tests have high power in low dimensions. Under sparse and decaying alternatives, the power decreases as $p$ grows, while the power increases for dense alternatives. 

\begin{figure}[tp]
	\centering
	\includegraphics[height=0.95\textheight]{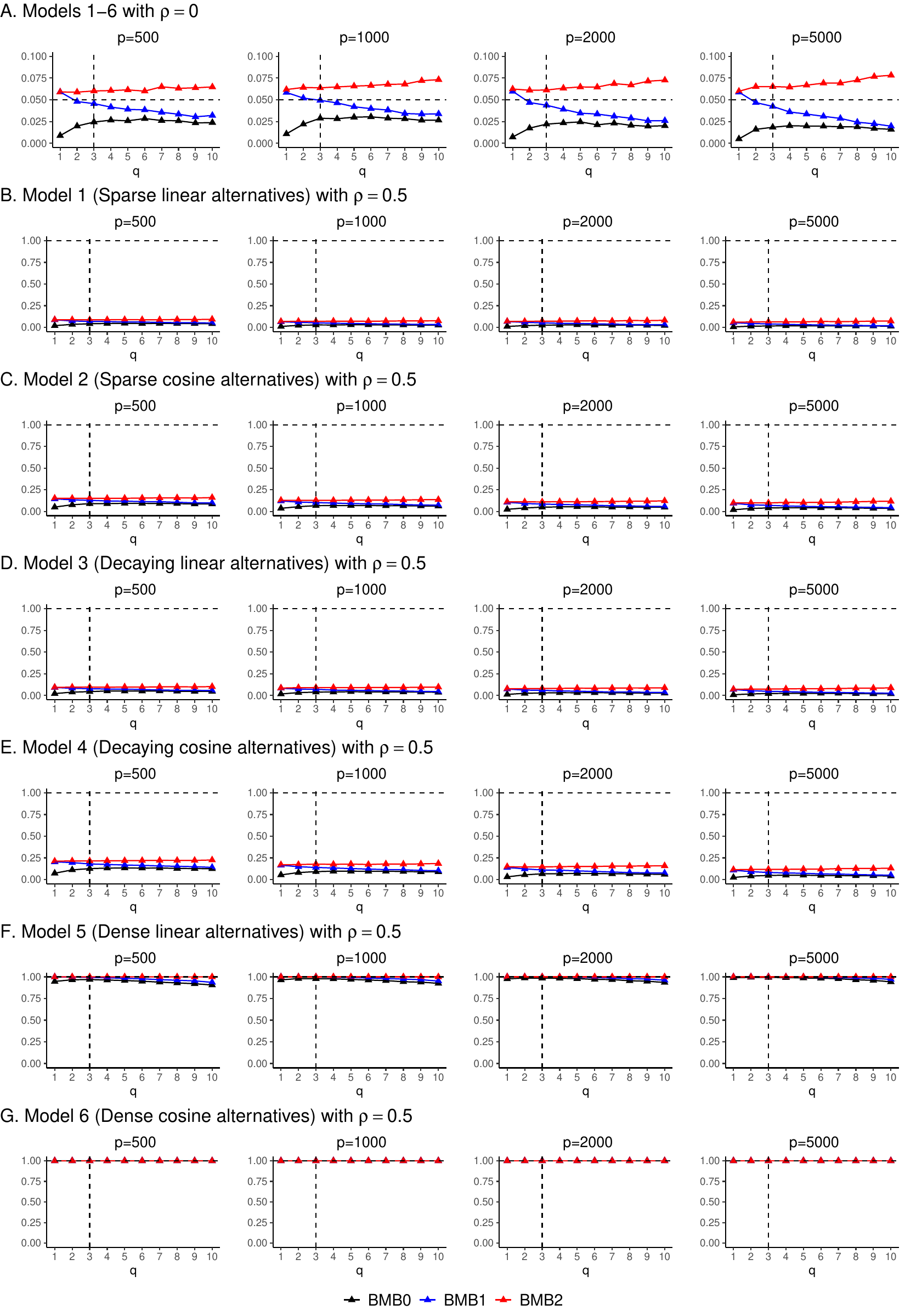}
	\caption{Rejection rates in Models 1--6 with $\tau=0$ under different choices of the big block size $q$.}
	\label{fig:SimI_tau0}
	\footnotesize{Notes: The tests have nominal level of 5\%. The null of independence holds if $\rho=0$, while $\tau$ denotes the correlation between variables $Y_1,\ldots, Y_p$ under the null. Results for sample size $n=500$, $B=499$ bootstrap replications, and $S=5,000$ Monte Carlo draws.}
\end{figure}

\subsection{Experiment 2: BMB's Size Control and Power}

Since, in Experiment 1, our test was seen to be fairly insensitive to the particular choice of block size $q$, we now analyse size control and power only for the optimal choice, which for $n=500$ is $q^*(n)=3$.

Figures~\ref{fig:power1} and~\ref{fig:power2} present the power curves of our test under linear and cosine alternatives, respectively. For each model, we consider the cases in which $Y_1,\ldots,Y_p$ are mutually independent ($\tau=0$) or dependent ($\tau=0.5$). In each graph, we show rejection frequencies of our test as we vary $\rho$. The null hypothesis corresponds to $\rho=0$, and as $\rho$ increases the violation of the null becomes larger.

\begin{figure}[tp]
	\centering
	\includegraphics[height=0.95\textheight]{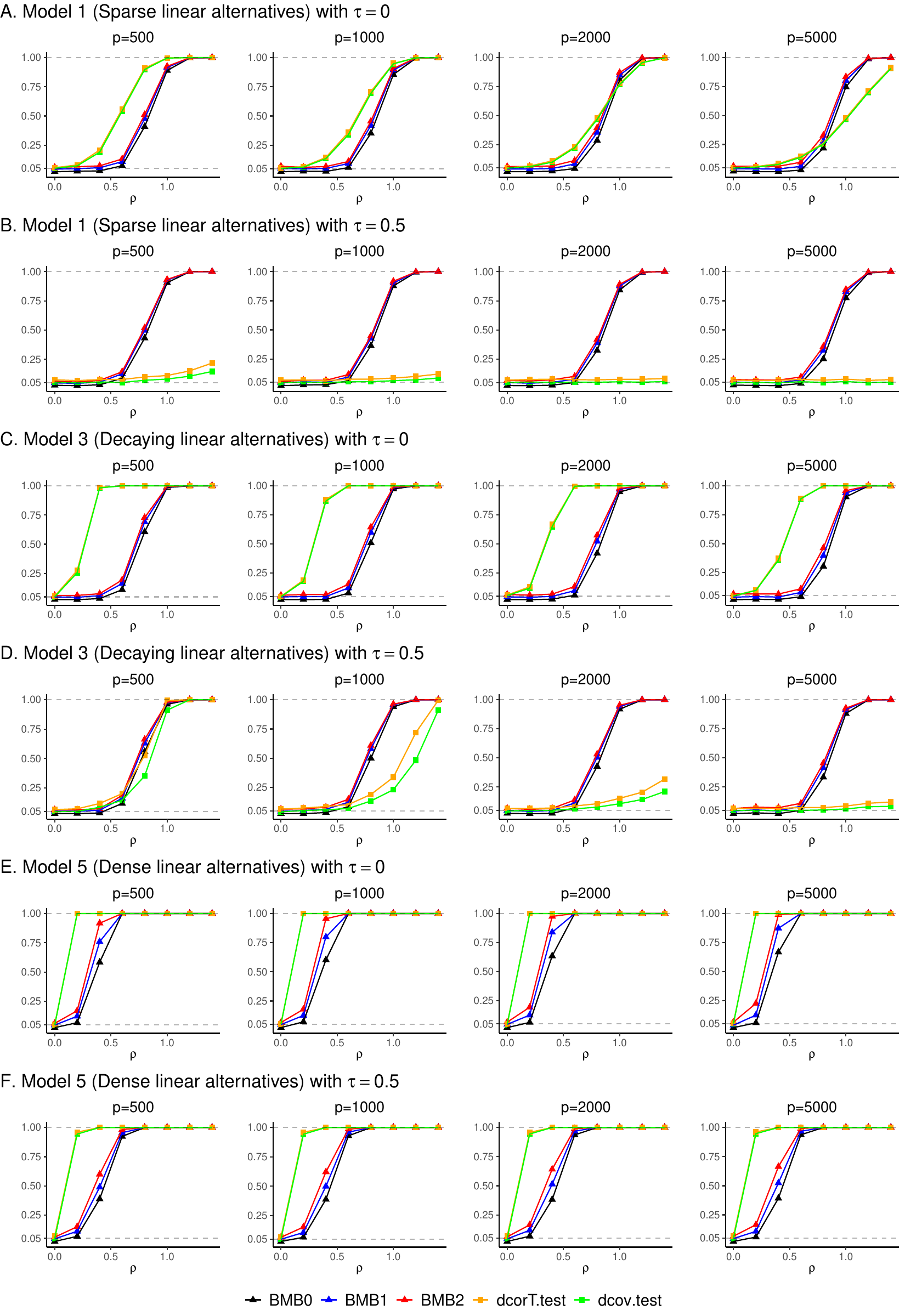}
	\caption{Power curves under {\bf linear alternatives} in high dimensions.}
	\footnotesize{Notes: The tests have nominal level of 5\%. Results for sample size $n=500$, $B=499$ bootstrap replications, and $S=5,000$ Monte Carlo draws.}
	\label{fig:power1}
\end{figure}

\begin{figure}[tp]
	\centering
	\includegraphics[height=0.95\textheight]{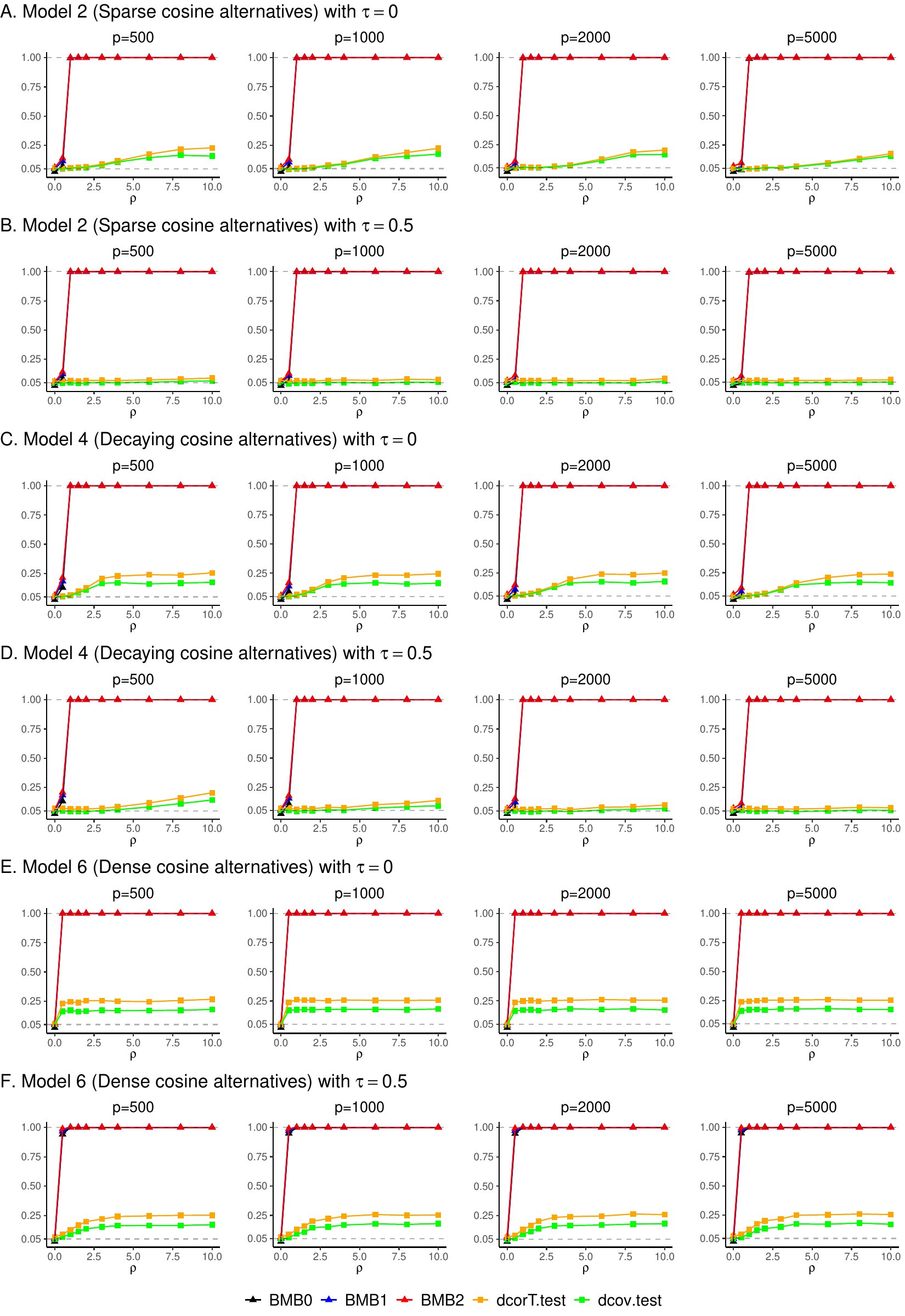}
	\caption{Power curves under {\bf cosine alternatives} in high dimensions.}
	\footnotesize{Notes: The tests have nominal level of 5\%. Results for sample size $n=500$, $B=499$ bootstrap replications, and $S=5,000$ Monte Carlo draws.}
	\label{fig:power2}
\end{figure}

First, the three variants of our test, BMB0, BMB1, and BMB2, control size across scenarios, including high-dimensional settings (large $p$) and those with dependence among $Y_1,\ldots,Y_p$ ($\tau=0.5$). 

Second, our test has power against all considered alternatives. Interestingly, comparing the power curves for $\tau=0$ and $\tau=0.5$, we see that the power of our test is not or only slightly affected by dependence among $Y_1,\ldots,Y_p$.

In Appendix~\ref{sec:additional_simulation_results}, we present analogous results for sample sizes $n=200$ and $n=1000$. While the power of the test increases with sample sizes, the qualitative conclusions remain the same.

\subsection{Experiment 3: Invalidity of Distance Covariance Tests}
\label{sec:invalidity_of_distance_covariance_tests}

As indicated above, the validity of existing tests dcorT.test, dcov.test, ZXZL, and ZZYS-agg.dcov have been established under the stronger null $H_0^{joint}$ that $(Y_1,\ldots,Y_p)\perp X$. We now investigate the size control properties of these tests when our null hypothesis $H_0\colon Y_j\perp X$ for $j=1,\ldots,p$ holds, but $H_0^{joint}$ is violated. 

To this end we consider a modified version of the data-generating processes Models~1--6, which are identical under the null, in which the dependence parameter $\tau$ depends on $X$. Specifically, we consider Model~1 with $\rho=0$ and instead of $(\epsilon_1,\ldots,\epsilon_p) \sim \mathcal N(0, \Sigma_\tau)$, we now consider $(\epsilon_1,\ldots,\epsilon_p) | X \sim \mathcal N(0, \Sigma_{\tau}(X))$, where $\Sigma_{\tau}(X)$ has diagonal elements equal to one and off-diagonal elements equal to $\tau(X) := \gamma (1+X)/2$. Since $\rho=0$, the marginal distributions of $Y_j$ are all independent of $X$, but the correlation parameter among $Y_1,\ldots,Y_p$ is a function of $X$. The parameter $\gamma$ governs the strength of that dependence. 

Figure~\ref{fig:rejection_vs_gamma} presents the rejection rates of all tests under the null hypothesis $H_0$ when $p=50$. Our test controls size across all values of $\gamma$, while the distance covariance tests dcorT.test and dcov.test do not and their rejection rates increase with $\gamma$. Interestingly, the ZXZL and ZZYS-agg.dcov tests also control size. Since they are based on marginal distance covariance and rank-based statistics, we conjecture that these tests are not only valid under $H_0^{joint}$, but also under $H_0$.

\begin{figure}[tp]
	\centering
	\includegraphics[width=0.8\textwidth]{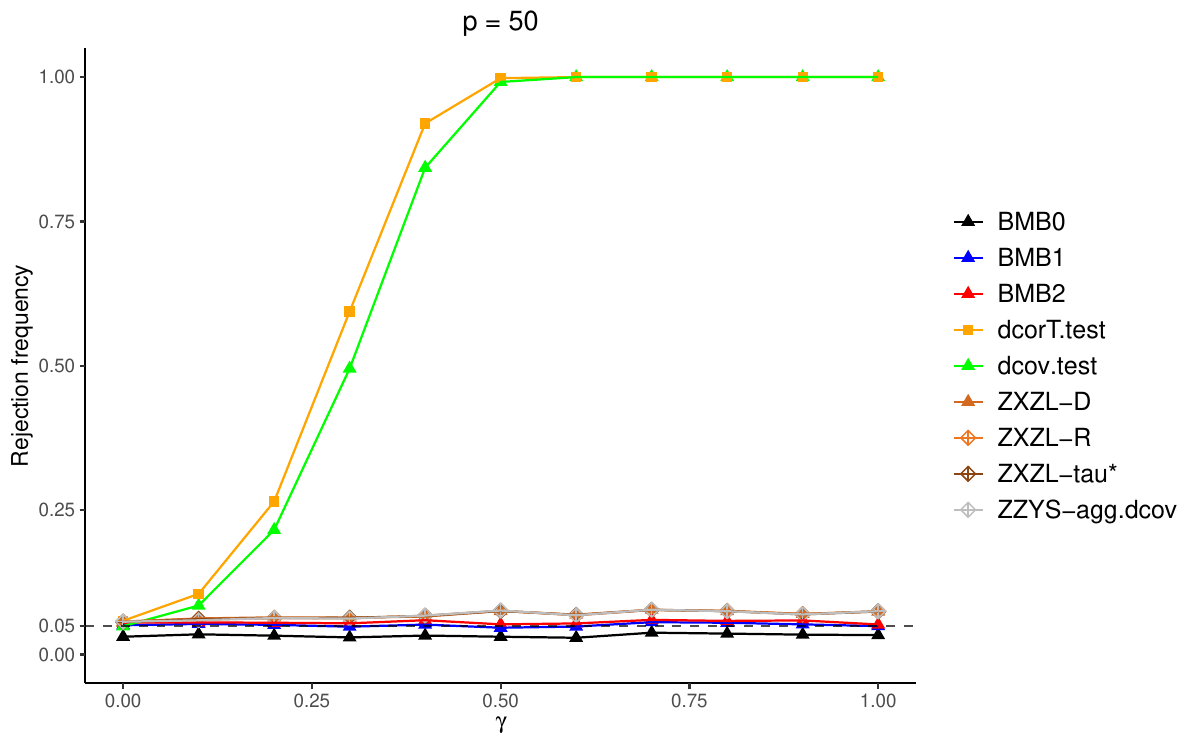}
	\caption{Rejection rates under the null when the copula of $Y_1,\ldots,Y_p$ given $X$ is not independent of $X$.}
	\footnotesize{Notes: The tests have nominal level of 5\%. Results for sample size $n=500$, $B=499$ bootstrap replications, and $S=5,000$ Monte Carlo draws.}
	\label{fig:rejection_vs_gamma}
\end{figure}

\subsection{Experiment 4: Comparison to Other Tests}

In this section, we compare our test to all other tests in the special case in which the copula of $Y_1,\ldots,Y_p$ given $X$ does not depend on $X$, both under the null and the alternative hypotheses. We divide this experiment into two sub-experiments because the ZXZL and ZZYS-agg.dcov tests are computationally so demanding that we were not able to run experiments in high dimensional settings ($p$ large). Therefore, we first consider an experiment in which we compare our test to the distance covariance tests dcorT.test and dcov.test for a range of values of $p$ up to very high dimensions. In a second experiment, we then compare our test to the ZXZL and ZZYS-agg.dcov tests in more moderate dimensions up to $p=500$.

\subsubsection{Experiment 4.1: Comparisons Including Very High-Dimensional Settings}

Figures~\ref{fig:power1} and~\ref{fig:power2} present the power curves of our test under linear and cosine alternatives, respectively. For each model, we consider the cases in which $Y_1,\ldots,Y_p$ are mutually independent ($\tau=0$) or dependent ($\tau=0.5$). In each graph, we show rejection frequencies of our test as we vary $\rho$. The null hypothesis corresponds to $\rho=0$, and as $\rho$ increases the violation of the null becomes larger.

Several interesting findings emerge from these figures. First, while our test is almost insensitive to dependence among $Y_1,\ldots,Y_p$, the distance covariance tests dcorT.test and dcov.test are not. For instance, comparing Panel A and B in Figure~\ref{fig:power1} for sparse alternatives, we see that their power may suffer substantially when $Y_1,\ldots,Y_p$ are dependent. A similar but less extreme pattern is observed under decaying linear alternatives (Panels C and D in Figure~\ref{fig:power1}). Under dense alternatives, under which these tests are known to perform particularly well, their power is less affected by dependence among $Y_1,\ldots,Y_p$.

Second, as expected, the distance covariance tests perform particularly well under dense alternatives, dominating the power curves of our tests. Our test performs particularly well under sparse and decaying alternatives, where it tends to be more powerful than the distance covariance tests, particularly so in high dimensions.

Third, considering the cosine alternatives (Figure~\ref{fig:power2}), we see that our test is more powerful than the distance covariance tests in all scenarios, including even the dense alternatives in which one would expect the distance covariance tests to perform well.

In Appendix~\ref{sec:additional_simulation_results}, we present analogous results for sample sizes $n=200$ and $n=1000$. While the power of all tests increases with sample sizes, the qualitative conclusions remain the same.

\subsubsection{Experiment 4.2: Comparisons Limited to Moderately High Dimensions}

In this section, we compare our test to the ZXZL and ZZYS-agg.dcov tests.\footnote{The tests ZXZL-D, ZXZL-R, ZXZL-$\tau^*$, and ZZYS-agg.dcov were run using the implementation of \citet{zhou2024rank} provided at \url{https://github.com/Yeqing-TJ/Rank-based-test-in-high-dimension} [Accessed on March 10, 2025].} Due to the high computational complexity of these additional tests, we limit the simulation setup to samples of size $n=200$ with the maximum of $p=200$ individual hypotheses.\footnote{Estimation of the variance of the rank-based indices of \citet{zhou2024rank} in one sample with $p=500$ and $n=500$ takes over 40 minutes using an Intel i7-1185G7 @ 3.00 GHz processor, which renders simulations for such settings infeasible. For comparison, our bootstrap test with $B=499$ takes less than a second to compute in this setting.}

Figures~\ref{fig:Sim_III_linear} and \ref{fig:Sim_III_cosine} show power comparisons analogous to Figures~\ref{fig:power1} and~\ref{fig:power2}, respectively. The broad patterns are similar to those in Figures~\ref{fig:power1} and~\ref{fig:power2} in that the ZXZL and ZZYS-agg.dcov tests exhibit similar performance to the distance covariance tests studied in Experiment 4.1. 

In the low-dimensional settings with linear alternatives (Figure~\ref{fig:Sim_III_linear}), the ZXZL and ZZYS-agg.dcov tests dominate our test when $p$ is small, there is no dependence among $Y_1,\ldots,Y_p$, or the alternatives are dense. Our test performs relatively better the larger $p$ and dominates the other tests when $p$ is sufficiently large (except for the dense alternatives). As in Experiment 4.1, our test's power curves are barely affected by dependence among $Y_1,\ldots,Y_p$, but the alternative tests' power curves decrease substantially when $Y_1,\ldots,Y_p$ are dependent; compare for instance Panels~A and B in Figure~\ref{fig:Sim_III_linear}. 

Under cosine alternatives (Figure~\ref{fig:Sim_III_cosine}), the ZXZL and ZZYS-agg.dcov tests have virtually no power, and our test is more powerful than the competitors in all scenarios.

\begin{figure}[tp]
	\centering
	\includegraphics[height=0.9\textheight]{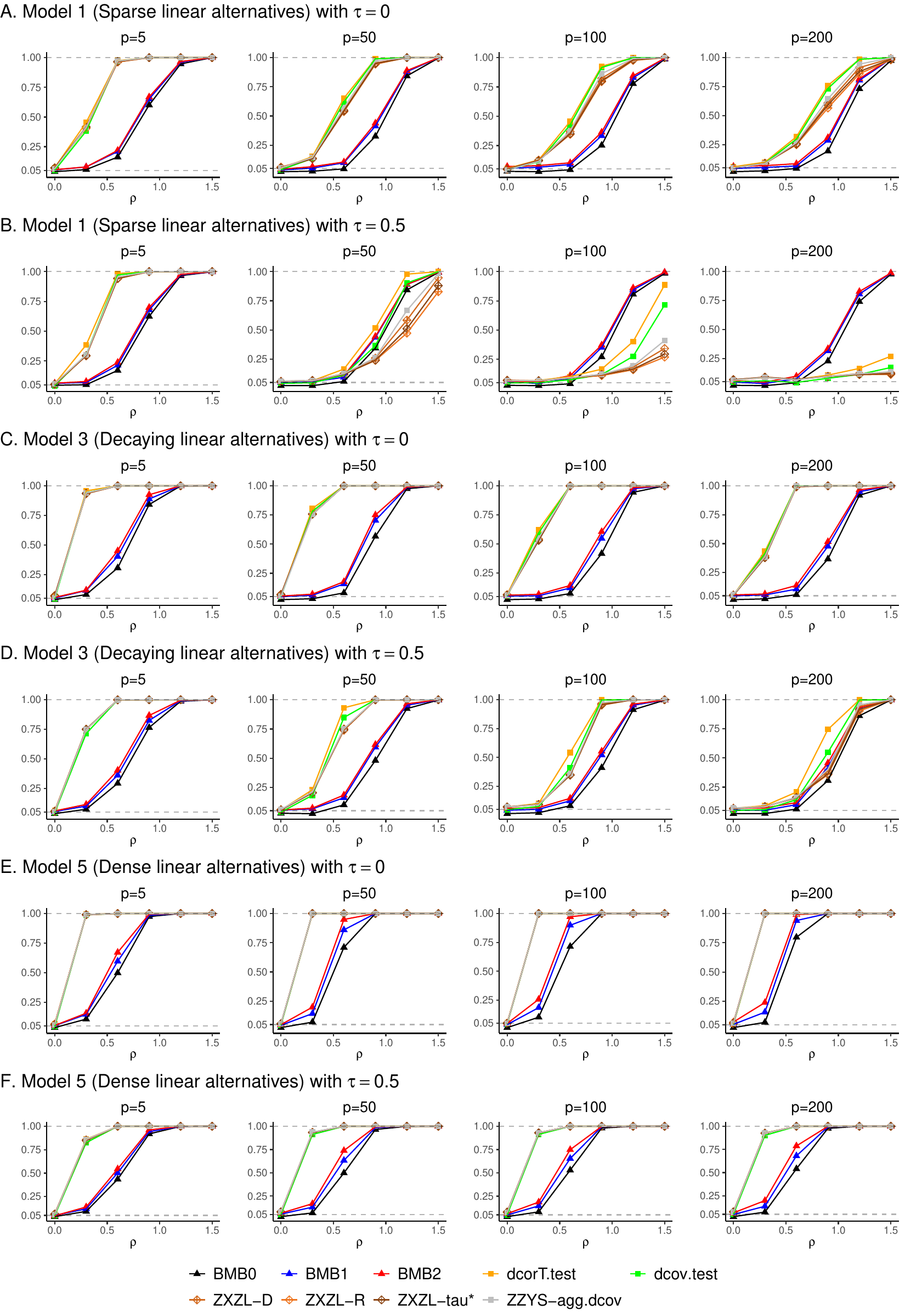}
	\caption{Power curves under {\bf linear alternatives} in moderately high dimensions.}
	\footnotesize{Notes: The tests have nominal level of 5\%. Results for $n=200$, $B=499$ bootstrap samples, and $S=1,000$ Monte Carlo draws.}
	\label{fig:Sim_III_linear}
\end{figure}

\begin{figure}[tp]
	\centering
	\includegraphics[height=0.9\textheight]{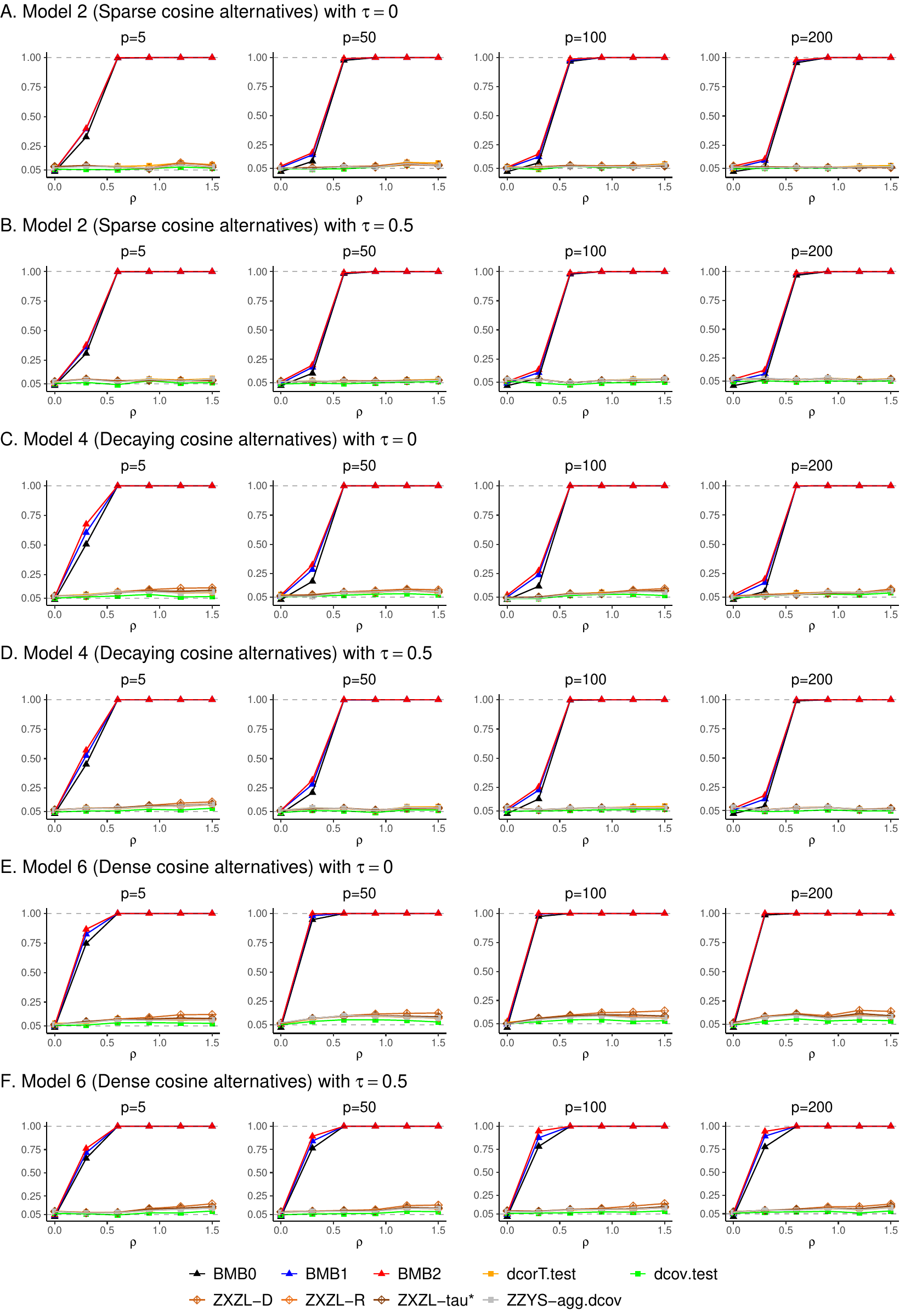}
	\caption{Power curves under {\bf cosine alternatives} in moderately high dimensions.}
	\footnotesize{Notes: The tests have nominal level of 5\%. Results for $n=200$, $B=499$ bootstrap samples, and $S=1,000$ Monte Carlo draws.}
	\label{fig:Sim_III_cosine}
\end{figure}

\subsection{Summary of Findings}

The simulation experiments reveal several interesting findings:

\begin{enumerate}
    \item Our test is fairly insensitive to the choice of block size $q$.
    \item Our test's performance is barely affected by the strength of dependence among $Y_1,\ldots,Y_p$, while the alternative tests' power deteriorates under dependence.
    \item Our test is relatively more powerful in higher dimensions and dominates the other tests when $p$ is sufficiently large, except when alternatives are dense.
    \item Our test is powerful against cosine alternatives, while the alternative tests are not.
    \item The distance covariance tests dcorT.test and dcov.test may not control size under the null hypothesis $H_0$.
    \item We did not find any evidence against validity of the ZXZL and ZZYS-agg.dcov tests under the null hypothesis $H_0$, but these tests are computationally much more demanding than ours and infeasible to simulate in high dimensions.
\end{enumerate}

\section{Empirical Application}\label{sec:empirical_application}

We illustrate the practical usefulness of our proposed test by revisiting the study conducted by \cite{hughes2009harmonics}, which investigates transcriptional oscillations from mouse liver, NIH3T3, and U2OS cells. The goal is to identify genes whose transcript levels oscillate during the cell cycle. These cycles are crucial as they play a key role in regulating metabolism and liver function, e.g., many liver genes operate on daily cycles, influencing vital processes such as detoxification, energy metabolism, and hormone regulation.

\cite{hughes2009harmonics} collected liver tissue samples from mice at hourly intervals over a 48-hour period, pooling samples from 3--5 mice at each point. Their gene-level expression data set is available from Gene Expression Omnibus (GEO). In this section, we focus on the liver data set (accession GSE11923). We extracted the data using the \texttt{GEOquery} and \texttt{BiocManager} R packages. The final data set contains $p=45,101$ genes. For each gene $j$, we observe $n=48$ transcript level measurements $Y_{j,1},\ldots,Y_{j,n}$ at different points in time, recorded in $X_1,\ldots,X_n$. 

The goal is to test whether gene transcriptions $Y_j$ are independent of the time of measurement $X$ and, in particular, to identify genes that violate independence. 

We apply our proposed bootstrap test based on the studentised test statistic BMB1 described in Remark~\ref{rem: studentization}, combined with the stepwise procedure developed in Section~\ref{sec:stepwise procedure}, both implemented via the R package \href{https://github.com/mauolivares/hdIndep}{\texttt{hdIndep}}. We set $\alpha=0.05$, the nominal level at which the family-wise error rate is to be controlled, and the number of bootstrap samples to $B=1,000$. We choose the optimal block size developed in Section~\ref{sec:choice of q}, which in this application is $q^*(n)=1$.

The stepdown procedure identifies $4,554$ genes that violate the independence hypothesis. These correspond to approximately $10.1\%$ of all gene transcripts. The stepdown procedure rejects $4,052$ hypotheses in the first and $502$ in the second step, demonstrating the power gains from employing multiple steps.

The original study by \cite{hughes2009harmonics} identified $3,667$ gene transcripts showing oscillatory behaviour, but based on a different methodology for testing a different hypothesis than ours.\footnote{The authors identify oscillatory patterns by testing for hidden periodicities in gene transcriptions using Fisher's G and \cite{straume2004dna} tests. These tests are designed to test the null that the data are generated by a Gaussian white noise process against the alternative that the data is generated by a Gaussian white noise with a deterministic sinusoidal component. The authors then obtain so-called q-values following \cite{storey2003statistical} so as to select gene transcripts whose q-values are less than $\alpha=0.05$.} Their procedure is guaranteed to control the false discovery rate at $\alpha=0.05$. Interestingly, our procedure finds more genes than the original study while at the same time providing stronger guarantees in the form of family-wise error rate control.

Among the genes identified by our method, $1,919$ were also identified in the original study, thus demonstrating substantial overlap. On the other hand, it singled out $2,635$ genes not identified in the original study (about $5.8\%$ of all gene transcripts), thus highlighting its ability to detect novel rhythmic gene activity. Figure~\ref{fig:six_plots} shows transcript level measurements of a random sample of six genes identified by our method, but not in the original study. A complete list of identified gene transcripts is available in the replication code (see replication vignette in accompanying R package \href{https://github.com/mauolivares/hdIndep}{\texttt{hdIndep}}).

\begin{figure}[t]
    \centering
    \includegraphics[width=\linewidth]{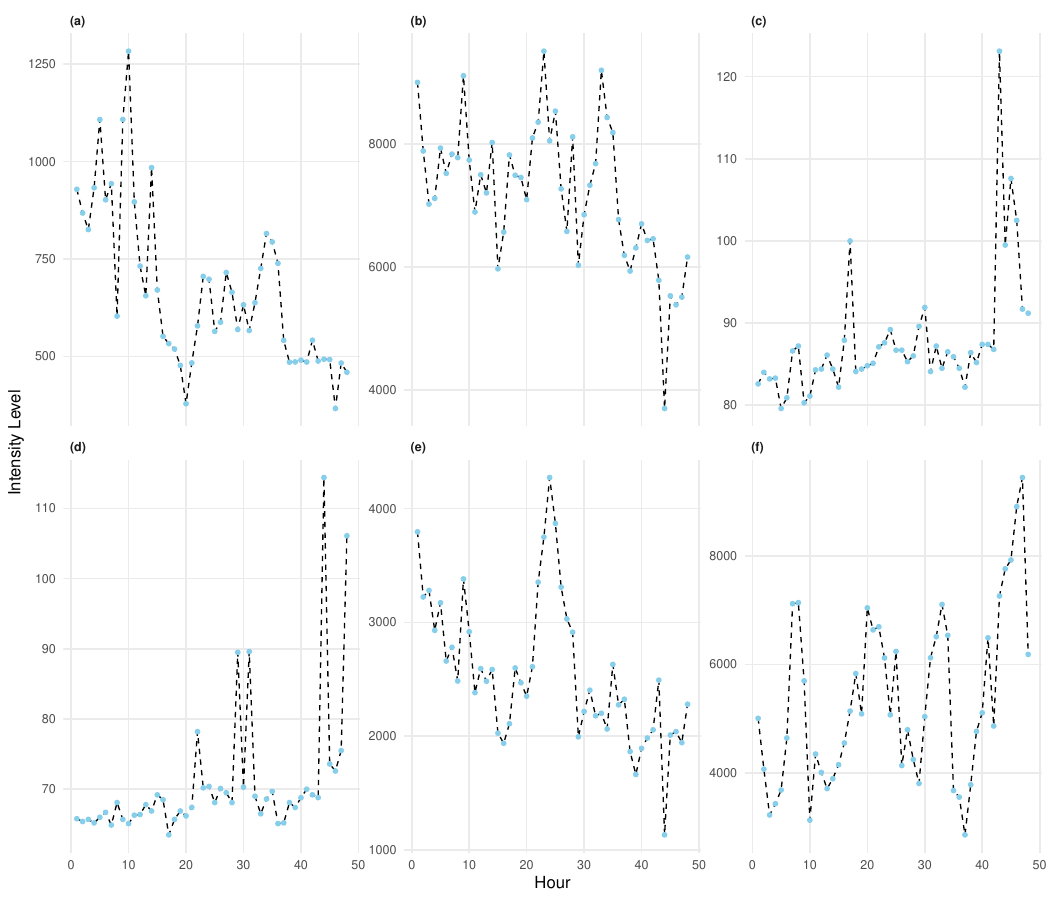}
    \caption{Visualisation of transcript levels of six genes uniquely identified by our proposed stepdown method in Section~\ref{sec:stepwise procedure}, with $B=1000$ Gaussian multiplier bootstrap samples, $\alpha=0.05$, and optimal $q^*(n)=1$. Gene symbols: (a) Hdac5, (b) Lrrc3, (c) Mrps2, (d) Uap1, (e) Mir5136, (f) Impg1.}
    \label{fig:six_plots}
\end{figure}

%%%%%%%%%%%%%%%%%%%%%%%%%%%%%%%%%%%%%%%%%%%%%%%%%%%%%%%%%%%%%%%%%%%%%%%%%%%%
%%%%%%%%%%%%%%%%%%%%%%%%%%%%%%%%%%%%%%%%%%%%%%%%%%%%%%%%%%%%%%%%%%%%%%%%%%%%

% Supplementary (conditionally compiled; see line 135 above)
\ifthenelse{\boolean{arxiv}}{
\clearpage
\newpage
\appendix

\section{Proofs}

\subsection{Proofs for Section~\ref{sec:size control}}

The proof of Theorem~\ref{thm: size control} relies on a few auxiliary lemmas. Define
\begin{equation}
	\Delta_1 \coloneqq \max_{1 \leq j \leq p}\abs{r_{j,n}}\quad \text{ and }\quad \Delta_2 \coloneqq \max_{1 \leq j \leq p} \frac{1}{mq}\sum_{k=1}^m ( \hat A_{j,k} - A_{j,k})^2.
\end{equation}
\begin{lemma}\label{lem: verify_Assumption}
	Suppose that Assumption~\ref{ass:continuity} and the null hypothesis $H_0$ hold, and
	\[
	\max\big\{ \log^{7/2}p, \sqrt{q}\log^{3/2}p\big\} \leq C_3 n^{1/2-c_3} 
	\]
	for some positive constants $c_3$ and $C_3$.
	Then, there exist positive constants $c$, $c_1$, $c_2$, $C$, $C_1$, and $C_2$ depending only on $c_3$ and $C_3$ such that
	\begin{itemize}
		\item[(i)]  $\Pr\left(\Delta_1 > C_1 n^{-c_1}/\sqrt{\log p}\,\right) \leq C n^{-c}$,
		\item[(ii)] $\Pr\left(\Delta_2 > C_2n^{-c_2}/\log^2 p\,\right) \leq C n^{-c}$.
	\end{itemize}
\end{lemma}

\begin{proof}
    The proof is divided into two parts, one for showing (i) and one for showing (ii). In the subsequent derivations, $c,\,c_1,\,c_2,\,C,\,C_1,\,C_2$ denote generic positive constants depending only on $c_3$ and $C_3$. Their values may change from place to place.

	\noindent\textit{Part (i).} 
    Note that $\hat{F}_{Y_j}(Y_{j,(i)})=\hat{F}_{U_j}(U_{j,i})$, where $U_{j,i} \coloneqq F_{Y_j}(Y_{j,(i)})$ and $\hat{F}_{U_j}(\cdot)$ is the empirical \cdf of $\{U_{j,i}\}_{i=1}^n$. Recall that under the null hypothesis of independence and Assumption~\ref{ass:continuity}, for each $j$, $\{U_{j,i}\}_{i=1}^n$ are independent random variables distributed uniformly on $[0,1]$. Following \citet{angus1995coupling}, we note that
	\begin{align*}
		\hat\xi_j  & =  1 - \frac{ 3n }{n^2-1}\sum_{i=1}^{n-1} \abs{ \hat{F}_{U_j}(U_{j,i+1})-\hat{F}_{U_j}(U_{j,i})} \\
		& = \frac{n}{n^2-1}\sum_{i=1}^{n-1} \left(1 - 3|U_{j,i+1} - U_{j,i}| -  3\abs{ \hat{F}_{U_j}(U_{j,i+1})-\hat{F}_{U_j}(U_{j,i})} + 3|U_{j,i+1} - U_{j,i}| \right)  + \frac{1}{n+1} \\
		& = \frac{n}{n^2-1} \sum_{i=1}^{n-1} \left(1 - 3|U_{j,i+1} - U_{j,i}|\right) + \frac{1}{n+1} \\
		& \quad -  \frac{n}{n+1} \left( 3 \int_0^1\int_0^1 \left( \abs{ \hat{F}_{U_j}(y)-\hat{F}_{U_j}(x) } - |y-x| \right) H(dx,dy) + 3 \tilde I_{j,n} \right)  ,
	\end{align*}
	where $H(x,y) \coloneqq xy$ and 
	\begin{align*}
		& \tilde I_{j,n}  \coloneqq \int_0^1 \int_0^1 \left( |\hat{F}_{U_j}(y)-\hat{F}_{U_j}(x) | - |y-x| \right)\big(\hat H_{j}(dx,dy)- H(dx,dy)\big),\\
		& \hat H_{j}(x,y)  \coloneqq \frac{1}{n-1} \sum_{i=1}^{n-1} \ind \{U_{j,i}\leq x, U_{j,i+1} \leq y\}.
	\end{align*}
	Furthermore, \citet{angus1995coupling} showed that
	\begin{align*}
		\int_0^1\int_0^1 \left( \abs{ \hat{F}_{U_j}(y)-\hat{F}_{U_j}(x) } - |y-x| \right)dydx 
		% & =  \int_0^1\int_x^1 \left( (\hat{F}_{U_j}(y) - y)- (\hat{F}_{U_j}(x)-x) \right) dydx 
		%\\ & \quad - \int_0^1\int_0^x \left( (\hat{F}_{U_j}(y) - y)- (\hat{F}_{U_j}(x)-x) \right) dydx \\
		%& = 2 \int_0^1  (2x-1) (\hat{F}_{U_j}(x)-x)  dx \\
		%& = \frac{2}{n} \sum_{i=1}^{n} \int_0^1 (2x-1) \big(\mathds{1}\{U_{j,i} \leq x\} - x \big) \\
		%& = \frac{2}{n} \sum_{i=1}^{n} \int_{U_{j,i}}^1 (2x-1) dx  - 2\int_0^1 (2x-1)x dx \\
		= \frac{1}{n} \sum_{i=1}^{n}2 U_{j,i}(U_{j,i}-1) - \frac{1}{3}.
	\end{align*}
	It follows that
	\begin{align*}
		\hat\xi_j = \frac{1}{n} \sum_{i=1}^{n-1} W_{j,i} - 3 \tilde I_{j,n} + u_{j,n},
	\end{align*}
	where $W_{j,i} \coloneqq 2 - 3 |U_{j,i+1} - U_{j,i}|  - 6 U_{j,i}(U_{j,i}-1)$ and $|u_{j,n}| < C/n$.

	To prove part~(i), it therefore suffices to show that
	\begin{equation}\label{eq:toshow}
		\Pr\left(\max_{1\leq j \leq p}|I_{j,n}| > C_1 n^{-c_1}/\sqrt{\log p}\,\right) \leq C n^{-c},\quad I_{j,n} \coloneqq \sqrt{n}\tilde I_{j,n}.
	\end{equation}
    To show that \eqref{eq:toshow} holds, first note that by the union bound, for any $\varepsilon>0$,
    \begin{equation}\label{eq:startingpoint}
        \Pr\left(\max_{1\leq j \leq p}|I_{j,n}| > \varepsilon \right) \leq p \max_{1\leq j \leq p} \Pr\left(|I_{j,n}| > \varepsilon \right).
    \end{equation}
    Next, we derive a bound on $\Pr\left(|I_{j,n}| > \varepsilon \right)$ that holds uniformly in $j$. Let $\hat B_j(x) = \sqrt{n}(\hat{F}_{U_j}(x)- x)$ and observe that
    \begin{equation}\label{eq:Isplit}
    \begin{split}   
    I_{j,n} & = \int_{0\leq y \leq x \leq 1} \big(  \hat B_j(x) - \hat B_j(y) \big)(\hat H_{j}(dx,dy)- H(dx,dy)) \\
        &\quad - \int_{0\leq x < y \leq 1} \big(  \hat B_j(x) - \hat B_j(y) \big) (\hat H_{j}(dx,dy)- H(dx,dy)).  
    \end{split}
    \end{equation}

    By the Koml\'{o}s-Major-Tusn\'{a}dy approximation, there exists a sequence of Brownian bridges $B_j$ \citep[][Theorem 3, Ch.~12]{shorack2009empirical} such that
    \begin{equation}\label{eq:Isplit5}
    \Pr\left(\norm{\hat B_j(x) - B_j(x)}_\infty > \varepsilon\right) \leq b^{\sss KMT} \exp\left(-c^{\sss KMT}(\sqrt{n} \varepsilon - a^{\sss KMT} \log n)\right)
    \end{equation}
    for some positive constants $a^{\sss KMT}$, $b^{\sss KMT}$, and $c^{\sss KMT}$ and all $\varepsilon>a^{\sss KMT}\log n/\sqrt{n}$. 

    By suitably adding and subtracting $B_j(x)$ and  $B_j(y)$ in \eqref{eq:Isplit}, we obtain that
    \begin{equation}\label{eq:Isplit2}
        |I_{j,n}| \leq 8 \norm{\hat B_j(x) - B_j(x)}_\infty +  \abs{\int_0^1\int_0^1 B_j(x) d\tilde h(x,y)}  + \abs{\int_0^1\int_0^1 B_j(y) d\tilde h(x,y)},
    \end{equation}
    where
    \[
        \tilde h(x,y) = \begin{cases}
            \hat H_{j}(x,y)- H(x,y) \text{ if } y\leq x,\\
            H(x,y)- \hat H_{j}(x,y) \text{ if } x< y.\\
        \end{cases}
    \]

    We now derive a bound on the second and third term on the right-hand side in \eqref{eq:Isplit2}.
    Fix a $\delta>0$ and consider a partition $\mathcal S$ of the interval $[0,1]$ into $\lceil 1/\delta \rceil$ intervals of length at most $\delta$. 
    Let $B^\delta_j$ be a process with piecewise-constant paths on the subintervals in $\mathcal S$ such that for any $S \in \mathcal S$ and $x \in S$, $B_j^\delta(x) = \inf_{y\in S} B_j(y)$. Recall that all $B_j$ are standard Brownian bridges on $[0,1]$, and therefore their distribution is the same for all $j$. By Levy's modulus of continuity theorem for the Brownian bridge \citep[][Theorem~1, Ch.~14]{shorack2009empirical}, with probability one, it holds that
    \[
    \norm{B^\delta_j(x) - B_j(x)}_\infty \leq 2 \sqrt{\delta \log(1/\delta)}.  
    \]
    It follows that
    \begin{align*}
        \left| \int_0^1\int_0^1 B_j(x) d\tilde h(x,y) \right| & = \Big| \int_0^1\int_0^1 \big(B_j(x)- B^\delta_j(x)\big) d\tilde h(x,y) +\int_0^1\int_0^1 B^\delta_j(x) d\tilde h(x,y)\Big|\\
        & \leq  4 \sqrt{\delta \log(1/\delta)} + \frac{2}{\delta} \norm{B_j}_\infty \norm{\hat H_{j}- H}_\infty.
    \end{align*}
    Using the union bound, we obtain that
    \begin{align*}
        \Pr\left( \frac{1}{\delta}\norm{B_j}_\infty \norm{\hat H_{j}- H}_\infty \geq \varepsilon \right) 
        & \leq  \Pr\left( \left\{\norm{B_j}_\infty\ge \varepsilon^{1/2}\delta^{1/2}n^{1/4}\right\} \cup \left\{ \norm{\hat H_{j}- H}_\infty \geq\varepsilon^{1/2}\delta^{1/2}n^{-1/4} \right\} \right) \\
        & \leq  \Pr\left(\norm{B_j}_\infty\ge \varepsilon^{1/2}\delta^{1/2}n^{1/4} \right) +  \Pr\left( \norm{\hat H_{j}- H}_\infty \geq\varepsilon^{1/2}\delta^{1/2}n^{-1/4} \right).
    \end{align*}
    By classical results for the Brownian bridge \citep[e.g.][Theorem~4.1]{adler1986tail}, 
    \[
        \Pr\left(\norm{B_j}_\infty\ge\lambda\right) \leq C \exp(-c\lambda^2)
    \]
    for all $\lambda > 0$. By Lemma~\ref{lemma:DKW2}, which is a simple implication of the two-dimensional version Dvoretsky-Kiefer-Wolfowitz inequality \citep{kiefer1958deviations},
    \[
        \Pr\left( \sqrt{n}\norm{\hat H_{j}- H}_\infty \geq \lambda \right) \leq C\exp(-c\lambda^2).
    \]
    for all $\lambda > 0$. In total, we obtain that
    \begin{equation}\label{eq:Isplit3}
        \Pr\left(\frac{1}{\delta} \norm{B_j}_\infty \norm{\hat H_{j}- H}_\infty \geq \varepsilon \right) \leq C^{\sss BH}\exp(-c^{\sss BH} \varepsilon\delta \sqrt{n}).
    \end{equation}
	for some positive constants $c^{\sss BH}$ and $C^{\sss BH}$.

    Using \eqref{eq:Isplit2} and the subsequent steps, we obtain that
    \[
        |I_{j,n}| \leq 8 \norm{\hat B_j(x) - B_j(x)}_\infty + 8 \sqrt{\delta \log(1/\delta)} + \frac{4}{\delta}\norm{B_j}_\infty \norm{\hat H_{j}- H}_\infty .
    \]
    Using the union bound, we have that
    \[
        \Pr\left(|I_{j,n}| > \varepsilon \right) \leq   \Pr\left( 8 \norm{\hat B_j(x) - B_j(x)}_\infty > \frac{\varepsilon}{3} \right) +    \Pr\left( 8 \sqrt{\delta \log(1/\delta)}  > \frac{\varepsilon}{3} \right)   + \Pr\left( \frac{4}{\delta}\norm{B_j}_\infty \norm{\hat H_{j}- H}_\infty  > \frac{\varepsilon}{3} \right) .
    \]

    Now, let $\varepsilon = C_1 n^{-c_1}/\sqrt{\log p}$ and take $\delta = C_\delta n^{-r}$ for some $r>1/7$. 
    Recall that, by assumption, $\log p \leq C_3 n^{a}$ for some $a < 1/7$. 
    With this choice of $r$, we can choose $C_\delta$ sufficiently large and $c_1$ sufficiently small such that $8 \sqrt{\delta \log(1/\delta)}  < \varepsilon/3$.
    Using \eqref{eq:Isplit5}, \eqref{eq:Isplit3}, we have that
    \begin{align*}
        \Pr\left(|I_{j,n}| >  C_1 n^{-c_1}/\sqrt{\log p} \right) & \leq \Pr\left(|I_{j,n}| >  C_1C_3^{-1/2} n^{-c_1-a/2} \right)\\
        & \leq  b^{\sss KMT} \exp(-c^{\sss KMT}(C_1C_3^{-1/2} n^{1/2-c_1-a/2}/24- a^{\sss KMT} \log n))  \\
        & \quad +  C^{\sss BH}\exp(-c^{\sss BH} C_1C_3^{-1/2}C_\delta n^{1/2-c_1-r-a/2}/12) \\
        & \leq  C\exp(-c n^{1/2-c_1-r-a/2}).
    \end{align*}
    Combining this result with the union bound in \eqref{eq:startingpoint}, we obtain that
    \begin{align*}
        \Pr\left(\max_{1\leq j \leq p}|I_{j,n}| > C_1 n^{-c_1}/\sqrt{\log p} \right) & \leq p\, C\exp(-c n^{1/2-c_1-r-a/2})\\
        & \leq C\exp(n^a -c n^{1/2-c_1-r-a/2}).
    \end{align*}
    Since $a < 1/7$, we can choose $c_1$ sufficiently small and $r$ sufficiently close to $1/7$ such that $a < 1/2 - c_1 -r -a/2$. Hence, we can choose positive constants $c$, $c_1$, $C$, and $C_1$ such that part~(i) holds.

    \noindent\textit{Part (ii).} 
    First, note that
    \[
        |\hat W_{j,l} - W_{j,l}| \leq C D_j, \quad D_j \coloneqq \sup_{t \in \mathbb{R}}\abs{\hat F_{j}(t) - F_{j}(t)}.
    \]
    Further,
    \begin{align*}
        (A_{j,k}-\hat A_{j,k})^2 = \left(\sum^{kq+(k-1)}_{l=(k-1)(q+1)+1} \left(\hat{W}_{j,l} - W_{j,l}\right) \right)^2 & \leq C q^2 D_j^2.
    \end{align*}
    It follows that
    \begin{align*}
        \Delta_2 & =  \max_{1 \leq j \leq p} \frac{1}{mq}\sum_{k=1}^m  (A_{j,k}-\hat A_{j,k})^2  \leq C q \max_{1\leq j \leq p} D_j^2.
    \end{align*}
    Thus, using the union bound and the DKW inequality, we have that for any $c_2,C_2>0$,
    \begin{align*}
    \Pr\big( \Delta_2 > C_2 n^{-c_2}/\log^2p\big) & \leq p  \max_{1\leq j \leq p} \Pr\left( C q D_j^2 > C_2 n^{-c_2}/\log^2p\right)\\
    &   \leq 2 p \exp\left( -c \frac{n^{1-c_2}}{q \log^2p} \right)\\
    &   \leq 2 \exp\left(\log p -c \frac{n^{1-c_2}}{q \log^2p} \right).
    \end{align*}
Note that 
\[
	\log p - c \frac{n^{1-c_2}}{q \log^2p} = \left( 1 - c \frac{n^{1-c_2}}{q \log^3p} \right) \log p\leq - C n^c, 
\]
provided that $c_2$ is chosen small enough. It then follows that $\Pr\big( \Delta_2 > C_2 n^{-c_2}/\log^2p\big) \leq Cn^{-c}$.
\end{proof}

\begin{lemma}\label{lem: CCKB1}
    Suppose that Assumption~\ref{ass:continuity} holds and there exist constants $C_1 > 0$ and $0<\gamma<1/4$ such that $(1/q)\log^2 p \leq C_1 n^{-\gamma}$ and
	$\max\{q\log^{1/2}p,\log^{3/2}p, \sqrt{q}\log^{7/2}(pn) \} \leq C_1 n^{1/2-\gamma}$.    
    Then there exist constants $c,C>0$ depending only on $\gamma$ and $C_1$ such that, under the null hypothesis $H_0$,
    \begin{equation}\label{eq: kolmogorov_distance}
        \sup_{t\in\mathbb{R}}\abs{\Pr\left(T_0\le t\right)-\Pr\left(Z_0\le t\right)}\le C n^{-c}~.
    \end{equation}
\end{lemma}

\begin{remark}
    Lemma~\ref{lem: CCKB1} is similar to Theorem~E.1 in \cite{chernozhukov2019inference}, but it differs from the latter in two ways: the mixing rate does not appear in our bound and we do not need to assume that the small blocks grow with the sample size. This difference stems from our derivations directly imposing the 1-dependence of $\{W_{j,i}\}_{i=1}^n$ rather than a general mixing condition.
\end{remark}

\begin{proof}
    In this proof, $c$, $C$ denote generic positive constants depending only on $\gamma$ and $C_1$; their values may change from place to place. We divide the proof into several steps. In the first three steps, we show that
    \begin{equation}\label{eq:Th1decomposition}
    - Cn^{-c} \leq \Pr(T_0\leq t) - \Pr\left(\max_{1\leq j\leq p} \sqrt{\frac{mq}{n}}\,V_j\leq t\right) \leq Cn^{-c} 
    \end{equation}
    for all $t \in \mathbb{R}$. For $\star \in \{+,-\}$, consider the following decomposition
    \begin{align*}
    \Pr(T_0\leq t) &  - \Pr\left(\max_{1\leq j\leq p} \sqrt{\frac{mq}{n}}\,V_j\leq t\right)\\
        & = \Pr(T_0 \leq t) - \Pr\left(\max_{1\leq j \leq p} \frac{1}{\sqrt{n}} \sum_{k=1}^m A_{j,k}\leq t \star  \frac{C}{n^c \sqrt{\log p} } \right) \tag{i} \\
        & \quad + \Pr\left(\max_{1\leq j \leq p} \frac{1}{\sqrt{n}} \sum_{k=1}^m A_{j,k}\leq t \star \frac{C}{n^c \sqrt{\log p} } \right) - \Pr\left( \max_{1\leq j \leq p} \sqrt{\frac{mq}{n}} V_j \leq t \star \frac{C}{n^c \sqrt{\log p} } \right) \tag{ii}\\
        & \quad + \Pr\left( \max_{1 \leq j \leq p}\sqrt{\frac{mq}{n}} V_j \leq t \star\frac{C}{n^c \sqrt{\log p} } \right) - \Pr\left( \max_{1 \leq j \leq p} \sqrt{\frac{mq}{n}} V_j \leq t  \right). \tag{iii}
    \end{align*}
    In steps 1--3, we show that each of the terms in (i)--(iii) is bounded in the supremum norm (w.r.t. $t$) by $Cn^{-c}$. To show the first and second inequality in \eqref{eq:Th1decomposition}, we use $\star = -$ and $\star = +$, respectively.

    \noindent\textbf{Step 1.} (Reduction to independence). We wish to show that
    \begin{align*}
        \Pr\left(\max_{1\leq j \leq p} \frac{1}{\sqrt{n}} \sum_{k=1}^m A_{j,k}\leq t - Cn^{-c}/\sqrt{\log p} \right) - n^{-c} 
        &\leq \Pr(T_0 \leq t) \\
        &\leq \Pr\left(\max_{1\leq j \leq p} \frac{1}{\sqrt{n}} \sum_{k=1}^m A_{j,k}\leq t + Cn^{-c}/\sqrt{\log p} \right) + n^{-c}.
    \end{align*}
    We only prove the second inequality; the first inequality follows from an analogous argument.
    Recall that $\sum_{i=1}^{n-1} W_{j,i} = \sum_{k=1}^m A_{j,k} + \sum_{k=1}^m B_{j,k} + R_j$, so that
    \[
    \abs{\max_{1\leq j \leq p} \sum_{i=1}^{n-1} W_{j,i} - \max_{1\leq j \leq p} \sum_{k=1}^m A_{j,k}} \leq 
    \max_{1\leq j \leq p} \abs{\sum_{k=1}^m B_{j,k}} + \max_{1\leq j \leq p}\abs{R_j}.
    \]
    Hence for every $\delta_1,\delta_2>0$,
    \begin{align*}
        \Pr(T_0\leq t) & \leq \Pr\left( \max_{1\leq j \leq p} \frac{1}{\sqrt{n}} \sum_{k=1}^m A_{j,k} \leq t + \delta_1 + \delta_2\right)  + \Pr\left(\frac{1}{\sqrt{n}}\max_{1\leq j \leq p} \abs{\sum_{k=1}^m B_{j,k}}> \delta_1 \right)\\
        & + \Pr\left(\max_{1\leq j \leq p} \abs{R_j}> \sqrt{n} \delta_2 \right)\\
        & \eqqcolon I + II + III.
    \end{align*}
    This inequality holds because the event on the {\small LHS} implies that at least one of the events on the {\small RHS} holds.

    Since $|W_{j,i}|\leq 2$ for all $j$ and $i$, we have $\abs{R_j}\leq 2(q+1)\,$ a.s., uniformly in $j$. Thus, take $\delta_2=4qn^{-1/2}$ ($\le Cn^{-c}/\sqrt{\log p}$ ) and conclude $III=0$. Moreover, for every $\varepsilon>0$, Markov's inequality and $\delta_1=\varepsilon^{-1}\E(n^{-1/2}\max_{1\leq j \leq p}|\sum_{k=1}^m B_{j,k}|)$ imply $II\leq \varepsilon$. This way, it remains to bound the magnitude of $\E(n^{-1/2}\max_{1\leq j \leq p}|\sum_{k=1}^m B_{j,k}|).$ To this end, we note that $|B_{j,k}|\leq 2$ a.s. uniformly in $j$ and $k$, and that $\max_{j}\sum_{k=1}^m\Var(B_{j,k})/m$ is bounded above since $\Var\left(B_{j,k}\right)=1/2$, independent of $j$. Thus, independence of $\{B_{j,k}\}_{k=1}^m$ and Lemma~\ref{lemma:CCKA3} imply
    \[
    \E\left(\max_{1\leq j \leq p} \abs{\frac{1}{\sqrt{n}} \sum_{l=1}^mB_{j,k}} \right) \leq K\left(\sqrt{\frac{\log p}{q}} + \frac{\log p}{\sqrt{n}}\right),
    \] 
    where $K$ is universal (here we have used the simple fact that $m/n\leq 1/q)$, so that the left side is bounded by $Cn^{-2c}/\sqrt{\log p}$ (by taking $c$ sufficiently small). The conclusion of this step follows from taking $\varepsilon = n^{-c}$ so that $\delta_1 \leq C n^{-c}/\sqrt{\log p}$.

    \noindent\textbf{Step 2.} (Normal approximation to the sum of independent blocks). We wish to show that 
    \[
    \sup_{t \in \mathbb R} \abs{\Pr\left(\max_{1\leq j \leq p} \frac{1}{\sqrt{n}} \sum_{k=1}^m A_{j,k}\leq t \right) - \Pr\left( \max_{1\leq j \leq p} \sqrt{\frac{mq}{n}} V_j \leq t \right)} \leq C n^{-c}.
    \]
    Since $\{A_{j,k}\}_{k=1}^m$ are independent, we may apply Corollary~2.1 in \citet{chernozhukov2013gaussian} (note that the covariance matrix of $\sqrt{mq/n}\,V$ is the same as that of $n^{-1/2} \sum_{k=1}^m A_{k,j}$). To this end, we seek to verify the conditions of the corollary applied to our case. Observe that 
    \[
    \frac{1}{\sqrt{n}}\sum_{k=1}^m A_{j,k} = \frac{1}{\sqrt{m}} \sum_{k=1}^m \frac{A_{j,k}}{\sqrt{n/m}},
    \]
    and $\sqrt{q}\leq \sqrt{n/m} \leq 2 \sqrt{q}$ (recall that $q+1\le (n-1)/2$). For any $1\le q\le n$, let
    \[
    \bar{\sigma}^2(q)=\max_{1\le j\le p}\max_I \text{Var}\left(\frac{1}{\sqrt{q}}\sum_{i\in I}W_{ij}\right)\quad\text{ and }\quad \underline{\sigma}^2(q)=\min_{1\le j\le p}\min_I \text{Var} \left(\frac{1}{\sqrt{q}}\sum_{i\in I}W_{ij}\right)~,
    \]
    where the $\max_I$ and $\min_I$ are taken over all $I\subset\{1,\dots,n\}$ of the form $I=\{i+1,\dots, i+q\}$. \citet{angus1995coupling} shows that $\E(W_{j,i})=0$, $\Var(W_{j,i})=1/2$, and $Cov(W_{j,i},W_{j,i+k})=-1/20$ for $k=1$ ($0$ otherwise). Therefore,
    \begin{align*}
    \text{Var}\left(\frac{1}{\sqrt{q}}\sum_{i\in I}W_{i,j}\right)&=\frac{1}{q}\left(\sum_{i\in I}\text{Var}(W_{i,j})+2\sum_{i=1}^{q-1}\,\text{Cov}(W_{j,i},W_{j,i+1})\right)=\frac{2}{5}+\frac{1}{10q}~,
    \end{align*}
    so that the variance in the previous display is a monotone decreasing function in $q$, independent of $j$ with maximal value $1/2$ at $q=1$ and bounded from below by $2/5$. Therefore $\bar{\sigma}^2(q)= \underline{\sigma}^2(q) = 2/5+1/(10q)$. Collect these facts to conclude that
    \[
    \frac{2}{5}\frac{1}{4} \leq \frac{\underline{\sigma}^2(q)}{4} \leq \Var\left(\frac{A_{j,k}}{\sqrt{n/m}}\right) \leq \bar{\sigma}^2(q)\leq \frac{1}{2},\quad1\le j\le p~,
    \]
    and $\abs{A_{j,k}/\sqrt{n/m}}\leq 2\sqrt{q}\,$ a.s. so the conditions of Corollary~2.1 (i) in \citet{chernozhukov2013gaussian} are verified with $B_n=2\sqrt{q}$, which leads to the assertion of this step (note that  $C^{-1}n^c\le n/(4q)\le m$).

    \noindent\textbf{Step 3.} (Anti-concentration). We wish to verify that, for every $\varepsilon>0$,
    \begin{equation}\label{eq}
    \sup_{t\in\mathbb{R}} \Pr\left(\abs{\max_{1\leq j \leq p} V_j-t}\leq \varepsilon \right) \leq C \varepsilon \sqrt{\max\{1 , \log(p/\varepsilon)\}}.
    \end{equation}
    Indeed, since $V$ is a normal random vector with 
    \[
    \frac{2}{5} \leq \underline{\sigma}^2(q)\leq \Var(V_j) \leq \bar\sigma^2(q)\leq \frac{1}{2},\quad 1\leq j \leq p~,
    \]
    the desired result follows by Corollary~1 in \citet{chernozhukov2015comparison} restated in Lemma~\ref{lemma:CCK2015cor1}. We apply this result with $\varepsilon=C/(2n^c \sqrt{\log p}) (mq/n)^{-1/2}$ and conclude that
    \[
    \sup_{t \in \mathbb{R}}\abs{    \Pr\left( \max_{1 \leq j \leq p}\sqrt{\frac{mq}{n}} V_j \leq t \star\frac{C}{n^c \sqrt{\log p} } \right) - \Pr\left( \max_{1 \leq j \leq p} \sqrt{\frac{mq}{n}} V_j \leq t  \right) }\leq Cn^{-c} \sqrt{\log n} \leq C n^{-c'}
    \]
    for some $c'$.

    \noindent\textbf{Step 4.} (Conclusion). By Steps~1--3, we have
    \[
    \sup_{t\in\mathbb{R}} \abs{\Pr(T_0\leq t) - \Pr\left(\max_{1\leq j\leq p} \sqrt{\frac{mq}{n}}\,V_j\leq t\right)} \leq Cn^{-c}.
    \]
    It remains to replace $\sqrt{(mq)/n}$ by 1 on the left side. Observe that
    \[
    1-\sqrt{\frac{mq}{n}} \leq 1 - \frac{mq}{n} \leq 1 - \left(\frac{n}{q+1}-1\right)\left(\frac{1}{n}\right) = \frac{1}{q+1} + \frac{1}{n},
    \]
    and the right side is bounded by $Cn^{-c}/\sqrt{\log p}$. With this $c$, by Markov's inequality,
    \[
    \Pr\left( \abs{\max_{1\leq j \leq p} V_j} > n^{c/2}\sqrt{\log p} \right) \leq C n^{-c/2},
    \]
    as $\E( \abs{\max_{1\leq j \leq p} V_j}\,)\leq C \sqrt{\log p}$, so that with probability larger than $1-Cn^{-c/2}$,
    \[
    \left(1-\sqrt{\frac{mq}{n}}\right) \abs{\max_{1\leq j \leq p} V_j} \leq C'n^{-c/2} \log^{-1/2}p.
    \]
    Apply the anti-concentration property of $\max_{1\leq j \leq p} V_j$ (see Step~3) to conclude 
    \[
    \sup_{t \in \mathbb{R}} \abs{\Pr\left( \max_{1\leq j \leq p} \sqrt{\frac{mq}{n}}\,V_j \leq t  \right) - \Pr\left(\max_{1\leq j \leq p} V_j \leq t \right)} \leq C n^{-c}.
    \]
    This finishes the proof of the lemma.
\end{proof}

\begin{lemma}\label{lem: CCKB2} Let $\varepsilon_1,\dots,\varepsilon_m$ be independent standard normal random variables, independent of the data $\mathbb{D}$. Suppose that there exist constants $C_1>0$ and $0 < \gamma < 1/2$ such that $q\log^{5/2}p \leq C_1 \,n^{1/2-\gamma}$. Then, under Assumption~\ref{ass:continuity} and the null hypothesis $H_0$, there exist constants $c,c',C,C'>0$ depending only on $\gamma$ and $C_1$ such that, with probability larger than $1-Cn^{-c}$,
\[
\sup_{t\in\mathbb{R}} \abs{\Pr\left(T_0^B\leq t\,\vert \mathbb{D} \right) - \Pr\left( Z_0 \leq t \right)} \leq C'n^{-c'}.
\]
\end{lemma}

\begin{proof}
    Here $c,c',C,C'$ denote generic positive constants depending only on $\gamma$; their values may change from place to place. By Theorem~2 in \citet{chernozhukov2015comparison}, the left side in the conclusion of the lemma is bounded by $C D^{1/3} \max\{1 , \log(p/D) \}^{2/3}$, where
    \[
    D = \max_{1\leq j, j'\leq p}\,\abs{\frac{1}{mq} \sum_{k=1}^m\bigl(A_{j,k} A_{j',k} - \E[A_{j,k} A_{j',k}]\bigr)}~. 
    \]  
    Hence, it suffices to prove that $\Pr(D > C'n^{-c'}\log^{-2}p)\leq Cn^{-c}$ with suitable $c,c',C,C'$. Observe that $\abs{A_{j,k}A_{j',k}}\leq 4q^2$ a.s. and $\E[(A_{j,k}A_{j',k})^2]\leq 4q^3 \,\bar\sigma^2(q)$, where $\bar{\sigma}^2(q)$ is given as in Step 2 in the proof of Lemma~\ref{lem: CCKB1}. Hence, by Lemma~\ref{lemma:CCKA3}, we have
    \begin{align*}
    \E[D] & \leq C\left( \sqrt{4m q^3 \bar\sigma^2(q)  } \frac{ \sqrt{\log p} }{mq} + 4q^2\frac{\log p}{mq}\right),\\
          & \leq C\left(  q \sqrt{ \frac{\log p}{mq} }  + q^2 \frac{\log p}{mq}\right)
    \end{align*}

    Since $mq > n/2$, it follows that
    \[
    \E[D] \leq C\left(q\sqrt{\frac{\log p}{n}} + q^2 \frac{\log p}{n}\right).
    \]
    Since $q\log^{5/2}p \leq C_1 \,n^{1/2-\gamma}$, the right side is bounded by $C'n^{-\gamma}\log^{-2}p$. The conclusion of the lemma follows from application of Markov's inequality.
\end{proof}

\begin{proof}[Proof of Theorem~\ref{thm: size control}]
    Let $c$ and $C$ denote generic positive constants depending only on $\gamma$ and $C_1$. Their values may change from place to place. The proof consists of several steps that rely on Lemmas~\ref{lem: verify_Assumption}, \ref{lem: CCKB1}, and \ref{lem: CCKB2}.

    \noindent{\bf Step 1:} 
    We will show that
    \begin{align}
        & \Pr\left(\abs{\hat{T} - T_0 } > \zeta_{n} \right) \leq C n^{-c}\label{eq: main_theorem_cond1} \\
        & \Pr\left( \Pr\left( \abs{ \hat T^B- T^B_0 }  > \zeta_{n} | \mathbb{D}  \right) > C n^{-c} \right) \leq C n^{-c}~.\label{eq: main_theorem_cond2} 
    \end{align}
    for some sequence $\zeta_n \leq Cn^{-c}/\sqrt{\log p}$ for sufficiently small $c>0$ and large $C>0$ depending only on the constants $c_1$, $c_2$, $C_1$, and $C_2$ from Lemma~\ref{lem: verify_Assumption}. Note the assumptions of Lemma~\ref{lem: verify_Assumption} are satisfied under the conditions of Theorem~\ref{thm: size control} because 
  	$\max\big\{ \log^{7/2}p, \sqrt{q}\log^{3/2}p\big\} \leq C_3 n^{1/2-c_3}$ follows from $\sqrt{q}\log^{7/2}(pn) \leq C_1 n^{1/2-\gamma}$ with $C_3=C_1$ and $c_3=\gamma$.
  
    To show~\eqref{eq: main_theorem_cond1}, observe that 
    \begin{align*}
    \abs{\hat{T} - T_0}&\le \max_{1\le j\le p}\abs{\hat{\xi}_j -\frac{1}{\sqrt{n}}\sum_{i=1}^{n-1}W_{j,i}}\le\max_{1\le j\le p}\abs{r_{j,n}}=\Delta_1~.
    \end{align*}
	The claim follows from Lemma~\ref{lem: verify_Assumption} by taking $\zeta_n \ge  C_1 n^{-c_1}/\sqrt{\log p}$. 
    
    To show~\eqref{eq: main_theorem_cond2}, note first that 
    \begin{align*}
	    \abs{ \hat T^B- T^B_0}\le \max_{1\le j\le p}\abs{\frac{1}{\sqrt{mq}}\sum_{k=1}^m\varepsilon_k\left(\hat{A}_{j,k}-A_{j,k}\right)}.
    \end{align*}
    Conditional on the data $\mathbb{D}$, the vector $\left(\frac{1}{\sqrt{mq}}\sum_{k=1}^m\varepsilon_k\big(\hat{A}_{j,k}-A_{j,k}\big)\right)_{1\leq j \leq p}$
    is normal with mean zero and all diagonal elements of the covariance matrix bounded by
    \begin{equation}
        \label{eq: main_theorem_cond3}
 	\max_{1\le j\le p}\frac{1}{mq}\sum_{k=1}^m \left(\hat{A}_{j,k}-A_{j,k}\right)^2.
    \end{equation}
	The last expression is bounded by $C_2n^{-c_2}/\log^2 p$ with probability larger than $1-Cn^{-c}$ by Lemma~\ref{lem: verify_Assumption}. Conditional on this event,
	\[
	\E\left[ \abs{ \hat T^B- T^B_0} | \mathbb{D} \right] \leq \sqrt{\frac{C_2n^{-c_2}}{\log^2 p}} \cdot 2\sqrt{2 \log p } \leq \frac{\zeta_n}{2}.
	\]
	Further, by Borell's inequality,
	\[
	\Pr\left( \abs{ \hat T^B- T^B_0 }  > \zeta_{n} \vert \mathbb{D} \right) \leq \exp\left( - \frac{\zeta^2_n \log^2 p}{8C_2n^{-c_2}}  \right) \leq \exp(-Cn^c) \leq C n^{-c}.
	\]
	This yields~\eqref{eq: main_theorem_cond2}. In the following derivations, $\zeta_n$ is as selected in this step.

    \noindent{\bf Step 2:} Let $c(\alpha)$ denote the ($1-\alpha$)-quantile of $Z_0$. In this step, we show that
    \begin{align}
        \Pr\left\{ \hat c(\alpha) \geq c(\alpha + C n^{-c} + 20\zeta_n \sqrt{\log p})\right\} &\geq 1-Cn^{-c},\label{eq: chat vs c ineq 1}\\
        \Pr\left\{ \hat c(\alpha) \leq c(\alpha - C n^{-c} - 20\zeta_n \sqrt{\log p})\right\} &\geq 1-Cn^{-c},\label{eq: chat vs c ineq 2}
    \end{align}
    and that, for any $\gamma\in(0,1-20\zeta_n\sqrt{\log p})$,
    \begin{equation}\label{eq: c ineq}
        c(\gamma+20\zeta_n \sqrt{\log p}) + \zeta_n \leq c(\gamma).
    \end{equation}
    We first show \eqref{eq: c ineq}. By Theorem~3 in \cite{chernozhukov2015comparison}, for any $t\in\mathbb{R}$ and any $\epsilon>0$,
    $$\Pr(|Z_0-t|\leq \epsilon ) \leq 4\epsilon (\sqrt{2\log p} + 1)/\sigma_V, $$
    where $\sigma_V^2 \coloneq \E[V_j^2]$. Since $\{A_{j,k}\colon j=1,\ldots,p,k=1,\ldots,m\}$ are i.i.d., $V_1,\ldots,V_p$ are also i.i.d., and
    $$\sigma_V^2 = \frac{1}{mq}\sum_{k=1}^m \E[A_{j,k}^2] = \frac{1}{q}\E[A_{j,k}^2] $$
    with $2/5 \leq E[A_{j,k}^2]/q \leq 1/2$, as shown in the proof of Lemma~\ref{lem: CCKB1}. Therefore,
    \begin{align*}
        &\Pr(Z_0 \leq c(\gamma+20\zeta_n \sqrt{\log p}) + \zeta_n) - \Pr(Z_0 \leq c(\gamma+20\zeta_n \sqrt{\log p}))\\
        &\qquad = \Pr(c(\gamma+20\zeta_n \sqrt{\log p}) \leq Z_0 \leq c(\gamma+20\zeta_n \sqrt{\log p}) + \zeta_n)\\
        &\qquad = \Pr( - \zeta_n/2\leq Z_0 -[c(\gamma+20\zeta_n \sqrt{\log p}) + \zeta_n/2]\leq \zeta_n/2)\\
        &\qquad \leq 2\zeta_n (\sqrt{2\log p} +1)/\sigma_V
    \end{align*}
    This implies
    \begin{align*}
        \Pr(Z_0 \leq c(\gamma+20\zeta_n \sqrt{\log p}) + \zeta_n) &\leq \Pr\left(Z_0 \leq c(\gamma+20\zeta_n \sqrt{\log p})\right) + 2\zeta_n (\sqrt{2\log p} +1)/\sigma_V\\
        &\leq \Pr\left(Z_0 \leq c(\gamma+20\zeta_n \sqrt{\log p})\right) + 8\zeta_n \sqrt{\log p} \cdot \frac{5}{2}\\
        &\leq 1-\gamma -20\zeta_n \sqrt{\log p} + 20\zeta_n \sqrt{\log p}\\
        &= 1-\gamma
    \end{align*}
    and the desired result in \eqref{eq: c ineq} follows.

	Next, we show \eqref{eq: chat vs c ineq 1}. For any $t\in\mathbb{R}$, by the union bound,
	\begin{align*}
		\Pr(\hat T^B\leq t|\mathbb{D}) &\leq \Pr(T_0^B \leq t + \zeta_n | \mathbb{D}) + \Pr(|\hat T^B - T_0^B|>\zeta_n | \mathbb{D})\\
		&\leq \Pr(Z_0 \leq t + \zeta_n ) + \rho_n^B + \Pr(|\hat T^B - T_0^B|>\zeta_n | \mathbb{D})
	\end{align*}
	where 
	$$\rho_n^B \coloneqq \sup_{t\in\mathbb{R}} \abs{\Pr\left(T_0^B\leq t\,\vert \mathbb{D} \right) - \Pr\left( Z_0 \leq t \right)}. $$
	Therefore, with probability at least $1-Cn^{-c}$,
	\begin{align*}
		\Pr(\hat T^B\leq c(\alpha + C n^{-c} + 20\zeta_n \sqrt{\log p})|\mathbb{D}) &\leq \Pr(Z_0 \leq c(\alpha + C n^{-c} + 20\zeta_n \sqrt{\log p}) + \zeta_n ) + Cn^{-c}\\
		&\leq \Pr(Z_0 \leq c(\alpha + C n^{-c} )) + Cn^{-c}\\
		&\leq 1- \alpha - C n^{-c} + Cn^{-c}\\
		&= 1-\alpha
	\end{align*}
	where the first inequality follows from Lemma~\ref{lem: CCKB2} and \eqref{eq: main_theorem_cond2}, and the second follows from \eqref{eq: c ineq}.
	
	Since $\Pr(\hat T^B\leq \hat c(\alpha)|\mathbb{D})=1-\alpha$, the inequality implies that with probability at least $1-Cn^{-c}$, $\hat c(\alpha) \geq c(\alpha + C n^{-c} + 20\zeta_n \sqrt{\log p})$, which establishes \eqref{eq: chat vs c ineq 1}. The second inequality \eqref{eq: chat vs c ineq 2} can be shown in an analogous fashion.

    \noindent{\bf Step 3:} This step combines results from Step 1, Step 2, and Lemmas~\ref{lem: CCKB1}-\ref{lem: CCKB2} to prove the statement of the theorem. Define
    $$\rho_n \coloneqq \sup_{t\in\mathbb{R}}\abs{\Pr\left(T_0\le t\right)-\Pr\left(Z_0\le t\right)}. $$
    Then:
    \begin{align*}
        \Pr\left(\hat T > \hat c(\alpha)\right) &\leq \Pr\left(T_0+ |\hat T - T_0| > \hat c(\alpha)\right)\\
            &\leq \Pr\left(\left\{T_0 + \zeta_n > \hat c(\alpha)\right\} \cup \left\{|\hat T -T_0| > \zeta_n\right\} \right)\\
            &\leq \Pr\left(T_0 + \zeta_n > \hat c(\alpha)\right) + \Pr\left(|\hat T -T_0| > \zeta_n\right)\\
            &= \Pr\left(T_0 + \zeta_n > \hat c(\alpha), \hat c(\alpha)\geq c(\alpha + C n^{-c} + 20\zeta_n \sqrt{\log p})\right)\\
                &\quad + \Pr\left(T_0 + \zeta_n > \hat c(\alpha)| \hat c(\alpha)< c(\alpha + C n^{-c} + 20\zeta_n \sqrt{\log p})\right)\times\\
                &\quad\quad \times \Pr\left( \hat c(\alpha) < c(\alpha + C n^{-c} + 20\zeta_n \sqrt{\log p})\right)\\
                &\quad+ \Pr\left(|\hat T -T_0| > \zeta_n\right)\\
            &\leq \Pr\left(T_0 + \zeta_n > c(\alpha + C n^{-c} + 20\zeta_n \sqrt{\log p})\right) + Cn^{-c}\\
            &\leq \Pr\left(T_0 > c(\alpha + C n^{-c} + 40\zeta_n \sqrt{\log p})\right) + Cn^{-c}\\
            &\leq \Pr\left(Z_0 > c(\alpha + C n^{-c} + 40\zeta_n \sqrt{\log p})\right) + \rho_n + Cn^{-c}\\
            &\leq \alpha + C n^{-c} + 40\zeta_n \sqrt{\log p} + \rho_n + Cn^{-c}\\
            &\leq \alpha + C n^{-c},
    \end{align*}
    where the third inequality follows from the union bound, the fourth inequality follows from \eqref{eq: main_theorem_cond1} and \eqref{eq: chat vs c ineq 1}, the fifth inequality follows from \eqref{eq: c ineq}, and the last inequality follows from Lemma~\ref{lem: CCKB1} and $40\zeta_n \sqrt{\log p}\leq C n^{-c}$.

    Similarly,
    \begin{align*}
        \Pr\left(\hat T > \hat c(\alpha)\right) &\geq \Pr\left(T_0 > \hat c(\alpha)+|\hat T - T_0|\right)\\
            &\geq \Pr\left(T_0  > \hat c(\alpha)+ \zeta_n\right) - \Pr\left(|\hat T -T_0| > \zeta_n\right)\\
            &\geq \Pr\left(T_0  > c(\alpha-Cn^{-c}-20\zeta_n\sqrt{\log p})+ \zeta_n\right) - Cn^{-c}\\
            &\geq \Pr\left(T_0  > c(\alpha-Cn^{-c}-40\zeta_n\sqrt{\log p})\right) - Cn^{-c}\\
            &\geq \Pr\left(Z_0  > c(\alpha-Cn^{-c}-40\zeta_n\sqrt{\log p})\right) - \rho_n - Cn^{-c}\\
            &\geq \alpha-Cn^{-c}-40\zeta_n\sqrt{\log p}  - Cn^{-c}\\
            &\geq \alpha-Cn^{-c},
    \end{align*}
    where the second inequality follows from the union bound, the third from \eqref{eq: main_theorem_cond1} and \eqref{eq: chat vs c ineq 2}, the fourth from \eqref{eq: c ineq}, the sixth from Lemma~\ref{lem: CCKB1}, and the last from $40\zeta_n \sqrt{\log p}\leq C n^{-c}$. This concludes the proof.
\end{proof}

\subsection{Proofs for Section~\ref{sec:consistency}}

\begin{proof}[Proof of Theorem~\ref{thm: consistency_theorem}]
    Let $v_n$ and $r_n$ be two deterministic sequences such that $v_n < r_n$. Using the union bound, we obtain that
    \begin{align*}
        \Pr\left( \hat T\leq \hat c(\alpha)\right) &= \Pr\left( \hat T- \hat c(\alpha)\leq r_n-v_n +v_n-r_n\right)\\
        &\leq \Pr\left( \hat T - \hat c(\alpha) \leq r_n - v_n\right) + \Pr\left(v_n\geq r_n \right)\\
        &\leq \Pr\left( \hat T \leq r_n \right) + \Pr\left(\hat c(\alpha) \geq v_n\right) + \Pr\left(v_n\geq r_n \right)~.
    \end{align*}
    To establish consistency, it suffices to show that
    \begin{align}
    \Pr\left(\hat c(\alpha) \ge v_n\right) &= o(1),\label{eq: cval rate}\\
    \Pr\left(\hat T\leq r_n \right) &= o(1)\label{eq: Tn rate}~.
    \end{align}

    We first establish~\eqref{eq: cval rate}. To ease notation, define, for $j=1,\dots,p$,
    \[
        \hat T^{B}_j=\frac{1}{\sqrt{mq}}\sum_{k=1}^{m}\varepsilon_k\hat{A}_{j,k}.
    \]
    Observe that the variance of $\hat T_j^B$ conditional on the data $\mathbb{D}$ is bounded from above by
    \begin{align*}
        \frac{1}{mq}\sum_{k=1}^m \hat{A}_{j,k}^2 &= \frac{1}{mq}\sum_{k=1}^m \left(\sum^{kq+(k-1)}_{l=(k-1)(q+1)+1} \hat{W}_{j,l}\right)^2\\
        &\leq \frac{1}{mq}\sum_{k=1}^m (2q)^2\\
        &= 4q,
    \end{align*}
    where the inequality follows from $|\hat W_{j,l}| \leq 2$. Then, by Proposition~1.1.3 in \cite{talagrand2003spin},
    \[
        \E\left[\hat T^B \vert \mathbb{D}\right] \leq 2\sqrt{2q\log p}~.
    \]
    A modification of Borell's Lemma (see Lemma~\ref{lem: Borell}) with $\sigma^2 = 4q $ and  $\lambda := \sqrt{8q\log(1/\alpha)}$ implies that
    \begin{equation*}
    \Pr\left\{\left.\max_{1\le j \le p} \hat T_j^B > \E\left[ \left.\max_{1\le j \le p} \hat T_j^B \;\right\vert\; \mathbb{D}\right] + \lambda\;\right\vert\; \mathbb{D}\right\} \leq \alpha~.
    \end{equation*}
    By the definition of $\hat c(\alpha)$, together with the anti-concentration inequality \citep[][Corollary 1]{chernozhukov2015comparison}, we have that
    \[
    \Pr\left(\left.\max_{1\le j \le p} \hat T_j^B > \hat c(\alpha) \;\right\vert\; \mathbb{D}\right) = \alpha~.
    \]
    Collecting these facts, we conclude that
    \begin{align*}
        \hat c(\alpha) &\leq \E\left[\left.\max_{1\le j \le p} \hat T_j^B \;\right\vert\; \mathbb{D}\right] + \sqrt{8q\log(1/\alpha)}\\
        &\leq 2 \sqrt{2q\log p} + 2\sqrt{2q\log(1/\alpha)}
    \end{align*}
    with probability one. This implies \eqref{eq: cval rate} with $v_n := 2 \sqrt{2q\log p} + 2\sqrt{2q\log(1/\alpha)}$.

    We now establish~\eqref{eq: Tn rate}. Note that under the alternative, at least one of the individual hypotheses $H_{0,j}:\,Y_j\bot X$, $1 \leq j \leq p$, fails. Let $j^*$ be the index of one of these hypotheses, so that $\xi_{j^*}>0$ (since $\xi_{j}$ is nonnegative for any $j$). Observe that $\hat T = \sqrt{n}\max_{1\le j\le p}\,\hat\xi_j \geq \sqrt{n}\, \hat{\xi}_{j^*}$. Since $r_n = o(n^{1/2})$ and $\xi_{j^*}>0$, we have that $r_n/\sqrt{n}-\xi_{j^*}<-\xi_{j^*}/2$ for sufficiently large $n$. Then, strong consistency of $\hat{\xi}_{j^*}$ \citep[][Theorem 1.1]{chatterjee2021new} implies
    \begin{align*}
      \Pr\left(\hat{T}>r_n\right)&\ge\Pr\left(\hat{\xi}_{j^*} -\xi_{j^*}>\frac{r_n}{\sqrt{n}}-\xi_{j^*}\right)\ge\Pr\left(\hat{\xi}_{j^*} -\xi_{j^*}>-\frac{\xi_{j^*}}{2}\right)=1-o(1)~.
    \end{align*}
    It remains to show that we can choose a sequence $r_n = o(n^{1/2})$ larger than $v_n$ defined above. Indeed, Assumption~\ref{ass:rates} implies that $\sqrt{q \log p}\leq C_1 n^{1/2-\gamma}$. Therefore, we can choose positive constants $c$ and $C$, such that $r_n \coloneqq C n^{1/2-c}$ satisfies $r_n = o(n^{1/2})$ and $r_n > v_n$, as required. This proves the claim in~\eqref{eq: Tn rate} and concludes the proof as a whole.
\end{proof}

\subsection{Proofs for Section~\ref{sec:choice of q}}

\begin{proof}[Proof of Lemma~\ref{lemma:VB}]
    Suppose that Assumption~\ref{ass:continuity} holds. In the subsequent derivations, we use the following values of moments of $W_{j,i}$ under the null $H_{0,j}$: 
    \begin{enumerate}[label=(\roman*)]
        \item $\E[W_{j,1}]=0$, 
        \item $\E[W_{j,1}^2]=1/2$, 
        \item $\E[W_{j,1}W_{j,2}] = -1/20$,
        \item $\E[W_{j,1}^4] = 3/5$,
        \item $\E[W_{j,1}^2W_{j,2}^2] = 23/70$,
        \item $\E[W_{j,1} W_{j,2} (W_{j,1}^2 + W_{j,2}^2)] = -3/28$,
        \item $\E[W_{j,1}W_{j,2}W_{j,3}(W_{j,1} + W_{j,2} + W_{j,3})] = -37/700$,
        \item $\E[W_{j,1}W_{j,2}W_{j,3}W_{j,4}] = 1/700$.
    \end{enumerate}
    Parts (i)--(iii) were shown by \citet{angus1995coupling}. Parts (iv)--(viii) were calculated using Mathematica.

    Note that
    \begin{align*}
        &\E[V_j^B] = \frac{1}{q}\E[A_{j,1}^2],\\
        &\E[A_{j,1}^2]  = q\Var(W_{j,1}) + 2(q-1)Cov(W_{j,1}, W_{j,2}) = \frac{2}{5}q + 1/10,
    \end{align*}
    which yields the first part of the lemma.

    Furthermore, by the independence of blocks, we have that
    \[
        \Var(V_j^B) = \frac{1}{m q^2}\Var(A_{j,1}^2) = \frac{1}{mq^2}\left(\E[A_{j,1}^4] - \E[A_{j,1}^2]^2 \right).
    \]
    For $q=1$, $\E[A_{j,1}^4] = \E[W_{j,1}^4] = 3/5$, and hence
    \[
        \E[A_{j,1}^4] - \E[A_{j,1}^2]^2 = 3/5 - 1/2^2 = 7/20.
    \]
    For $q=2$, 
    $\E[A_{j,1}^4] = 2 \E[W_{j,1}^4] + 4\E[W_{j,1} W_{j,2} (W_{j,1}^2 + W_{j,2}^2)] + 6\E[W_{j,1}^2 W_{j,2}^2] = 96/35$,
    and hence
    \begin{align*}
        \E[A_{j,1}^4] - \E[A_{j,1}^2]^2 = 96/35 - (9/10)^2.
    \end{align*}
    For $q\ge 3$, a tedious calculation shows that $\E[A_{j,1}^4] = \frac{48}{100} q^2 + \frac{102}{175} q - \frac{111}{350}$, which combined with the expression for $\E[A_{j,1}^2]$ then leads to the desired result.
\end{proof}

\subsection{Proofs for Section~\ref{sec:stepwise procedure}}

\begin{proof}[Proof of Theorem~\ref{thm: FWER control}]
    The result follows from \citet{romano2005exact} if
    \begin{align}
        \hat c(\alpha;I') \leq \hat c(\alpha;I'') \text{ whenever } I' \subset I'', \label{eq:RWcond1} \\
            \sup_{P \in \mathbf{P}_{\gamma,C_1}} P\left(\hat T(J(P)) > \hat c(\alpha; J(P)) \right) \leq \alpha + o(1). \label{eq:RWcond2}
    \end{align}
    Condition \eqref{eq:RWcond1} holds by construction. By inspecting the proof of Theorem~\ref{thm: size control}, we see that \eqref{eq:RWcond2} holds.
\end{proof}

\section{Auxiliary Results}
\label{sec: aux results}

% ----------------------------------------------------------------- %
\begin{lemma}[Corollary~1 in \citet{chernozhukov2015comparison}] \label{lemma:CCK2015cor1}
    Let $(X_1,\ldots,X_p)^T$ be a centred Gaussian random vector in $\mathbb{R}^p$ with $\sigma_j^2 \coloneq \E[X_j^2]>0$ for all $1\leq j \leq p$. Let $\underline{\sigma} \coloneqq \min_{1\leq j \leq p}\sigma_j $ and $\bar{\sigma} \coloneqq \max_{1\leq j \leq p}\sigma_j$. Then for every $\epsilon>0$,
    \[
    \sup_{x \in \mathbb{R}} \Pr\left(  \abs{ \max_{1\leq j \leq p}X_j  -x } \leq \epsilon \right) \leq C \epsilon \sqrt{\max\{1 , \log(p/\epsilon)\} },
    \]
    where $C >  0$ depends only on $\underline{\sigma}$ and $\bar\sigma$. When $\sigma_j$ are all equal, $\log(p/\epsilon)$ on the right side can be replaced by $\log p$.
\end{lemma}

\begin{lemma}[Lemma 8 in \cite{chernozhukov2015comparison}]\label{lemma:CCKA3}
    Let $X_1,\dots,X_n$ be independent random vectors in $\mathbb R^p$ with $p \ge 2$. Define $M\coloneqq \max_{1\leq i \leq n} \max_{1\leq j \leq p}\abs{X_{ij}}$ and $\sigma^2\coloneqq \max_{1 \leq j \leq p} \sum_{i=1}^n \E[X_{ij}^2] $. Then
    \[
        \E\left(\max_{1\leq j \leq p}\abs{ \sum_{i=1}^n\left(X_{ij}-\E[X_{ij}]\right) } \right) \leq K\left(\sigma \sqrt{\log p} + \sqrt{\E[M^2]}\log p\right),    
    \]
    where $K$ is a universal constant.
\end{lemma}

% ----------------------------------------------------------------- %
\begin{lemma}[Multivariate DKW inequality, \citet{kiefer1958deviations} ]\label{lemma:DKW}
    Let $X_1,\dots,X_n$ be independent random variables in $\mathbb{R}^k$ with \cdf $F(\cdot)$, and let $\hat F$ be their empirical \cdf. There exist constants $c_k$ and $c_k'$ such that for all $\varepsilon>0$ and $n\in\mathbb{N}$ it holds that 
    \[
    \Pr\left(\sup_{t \in \mathbb{R}^k}\abs{F_n(t) - F(t)} \ge \varepsilon \right) \leq c_k\,\exp\left(-c_k'n\varepsilon^2\right)\quad \text{ for any }\varepsilon>0.
    \]
\end{lemma}

The following DKW-type inequality for 1-dependent data is a straightforward corollary of Lemma~\ref{lemma:DKW}.

% ----------------------------------------------------------------- %
\begin{lemma}\label{lemma:DKW2}
    Let $X_1,\ldots,X_n$ be a sequence of 1-dependent, identically distributed random variables in $\mathbb R^k$. Then the conclusion of Lemma~\ref{lemma:DKW} holds.
\end{lemma}

\begin{proof}
    Suppose that $n=2N$, and let $F_n'$ be the empirical \cdf calculated based on observations with odd indices $i$, and $F_n''$ based on the even indices.   Then for any $\varepsilon>0$, using the union bound, we obtain that
    \begin{align*}
        \Pr\left( \sup_{t\in\mathbb{R}^k}\abs{F_n(t) - F(t)} > \varepsilon \right) & = \Pr\left( \sup_{t\in\mathbb{R}^k}\abs{(F'_n(t) - F(t)) + (F''_n(t) - F(t))} > 2\varepsilon \right) \\
        & \leq \Pr\left( \sup_{t\in\mathbb{R}^k}\abs{F'_n(t)- F(t)} > \varepsilon \right) +  \Pr\left( \sup_{t\in\mathbb{R}^k}\abs{F''_n(t)- F(t)} > \varepsilon \right) \\
        & \leq 2 c_k e^{-c_k'n \varepsilon^2/2}.
    \end{align*}
An analogous argument holds for odd $n.$    
\end{proof}

% ----------------------------------------------------------------- %
\begin{lemma}[Modified Borell's Inequality]\label{lem: Borell}
    Let $X_1,\ldots,X_p$ be mean-zero normal random variables with $\sigma^2 := \max_{1\le i \le p}Var(X_i)$ finite and nonzero. Then, for any $\lambda>0$,
    \[
    \Pr\left(\max_{1\le i\le p} X_i > \E\left[ \max_{1\le i\le p} X_i\right] + \lambda\right) \leq e^{-\lambda^2/(2\sigma^2)}~.
    \]
\end{lemma}

\begin{proof}
    Let $Z:=(Z_1,\ldots,Z_p)$ be a vector independent standard normal random variables and let $A$ be the symmetric square-root of $Var(X)$, $X:=(X_1,\ldots,X_p)$. For any $z\in\mathbb{R}^p$, let
    $$f(z) := \frac{1}{\sigma} \max_{i=1,\ldots,p} (Az)_i $$
    and note that $f$ is a Lipschitz function with Lipschitz constant $1$: for any $z_1,z_2\in\mathbb{R}^p$
    \begin{align*}
        |f(z_1)-f(z_2)|&= \frac{1}{\sigma}\left|  \max_{1\le i \le p} (Az_1)_i -  \max_{i=1,\ldots,p} (Az_2)_i\right|\\
        &\leq \frac{1}{\sigma}  \max_{1\le i \le p} \left|(A(z_1 -  z_2))_i\right|\\
        &\leq \frac{1}{\sigma}  \max_{1\le i \le p} \|A_{i\cdot}\|_{\infty}\; \|z_1 -  z_2\|_1\\
        &\leq \frac{1}{\sigma}  \left(\max_{1\le i \le p} A_{ii}^2\right)^{1/2}\; \|z_1 -  z_2\|_1\\
        &= \|z_1 -  z_2\|_1
    \end{align*}
    where the second inequality follows from H\"older's inequality. We can therefore apply \citet[][Lemma~A.2.2]{vanderVaart1996weak} to obtain
    \begin{align*}
        \Pr\left(\max_{1\le i \le p} X_i > \E\left[ \max_{1\le i\le p} X_i\right] + \lambda \right) &= \Pr\left(f(Z) - \E[ f(Z)] >  \frac{\lambda}{\sigma} \right) \leq e^{-\lambda^2/(2\sigma^2)}~,
    \end{align*}
    as desired.
\end{proof}

\section{Additional Simulation Results}\label{sec:additional_simulation_results}

\subsection{Additional Results for Experiment 1}
Figure~\ref{fig:SimI_tau05} complements the results from the main text by considering the case where $Y_1,\ldots,Y_p$ have baseline correlation. The results are very similar to those presented in Figure~\ref{fig:SimI_tau0}, which shows that our procedure is not sensitive to dependence between the individual test statistics.

\begin{figure}
	\centering
	\includegraphics[height=0.95\textheight]{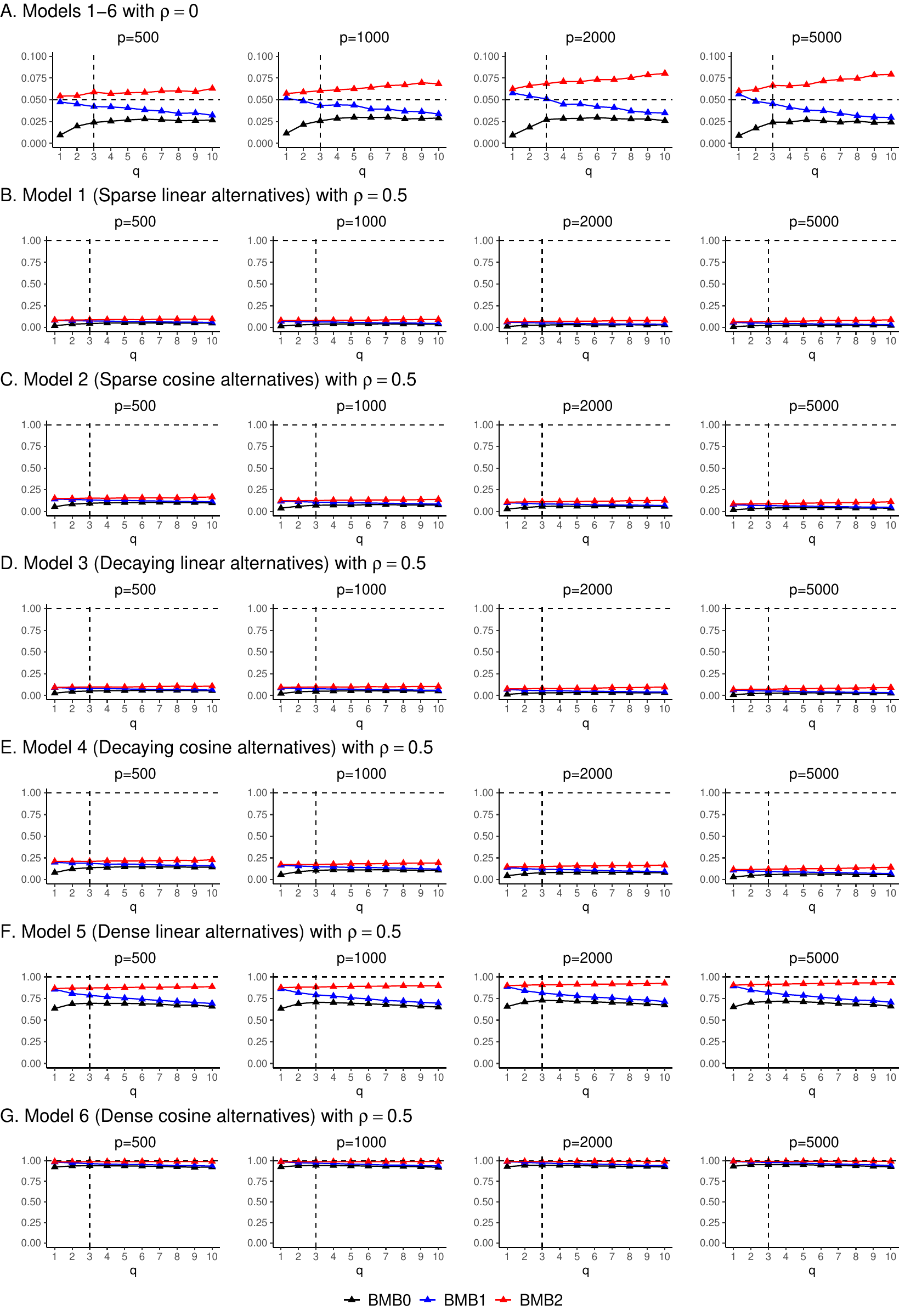}
	\caption{Rejection rates in Models 1--6 with $\tau=0.5$ under different choices of the big block size $q$.}
	\label{fig:SimI_tau05}
	\footnotesize{Notes: The tests have nominal level of 5\%. The null of independence holds if $\rho=0$, while $\tau$ denotes the correlation between variables $Y_1,\ldots, Y_p$ under the null. Results for sample size $n=500$, $B=499$ bootstrap replications, and $S=5,000$ Monte Carlo draws.}
\end{figure}

\subsection{Additional Results for Experiments 2 and 4}
\label{app:additional_results_for_simulation_ii}
In this section, we extend the results of Experiments 2 and 4 from the main text by including samples of size $n=200$ and $n=1000$. 
The results are presented in Figures~\ref{fig:SimII-start}--\ref{fig:SimII-end}. As expected, the rejection rates (weakly) increase with the sample size for any given combination of $p$ and $\rho$, but the qualitative patterns remain similar to those discussed in the main text.

\begin{figure}
	\centering
	\includegraphics[height=0.95\textheight]{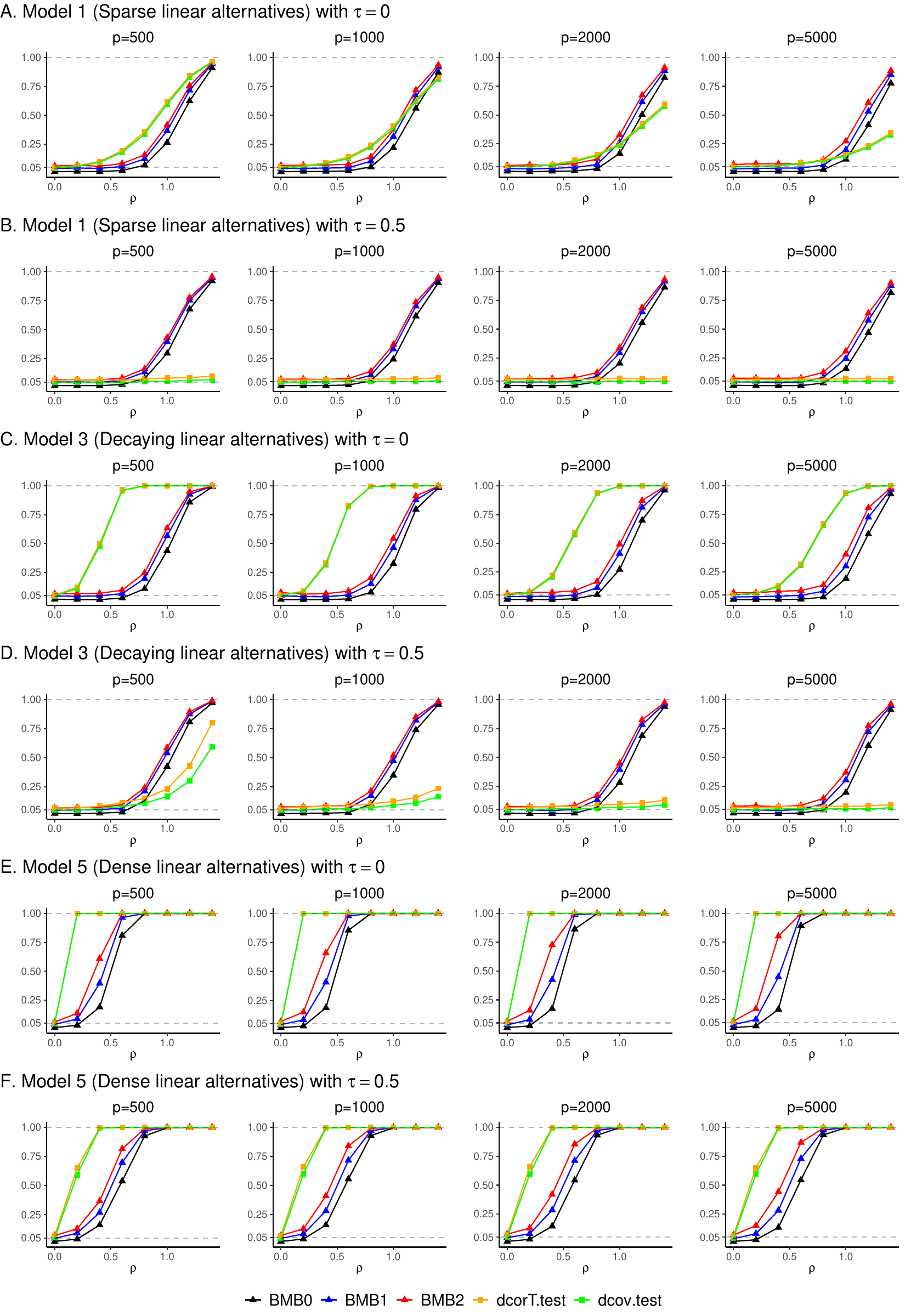}
	\caption{Experiments 2 and 4 for \textbf{linear alternatives} with $n=200$.}
	\footnotesize{Notes: The tests have nominal level of 5\%. Results are based on $S=5,000$ Monte Carlo draws. The BMB tests and the distance covariance test used $B=499$ bootstrap samples/replicates.}
	\label{fig:SimII-start}
\end{figure}

\begin{figure}
	\centering
	\includegraphics[height=0.95\textheight]{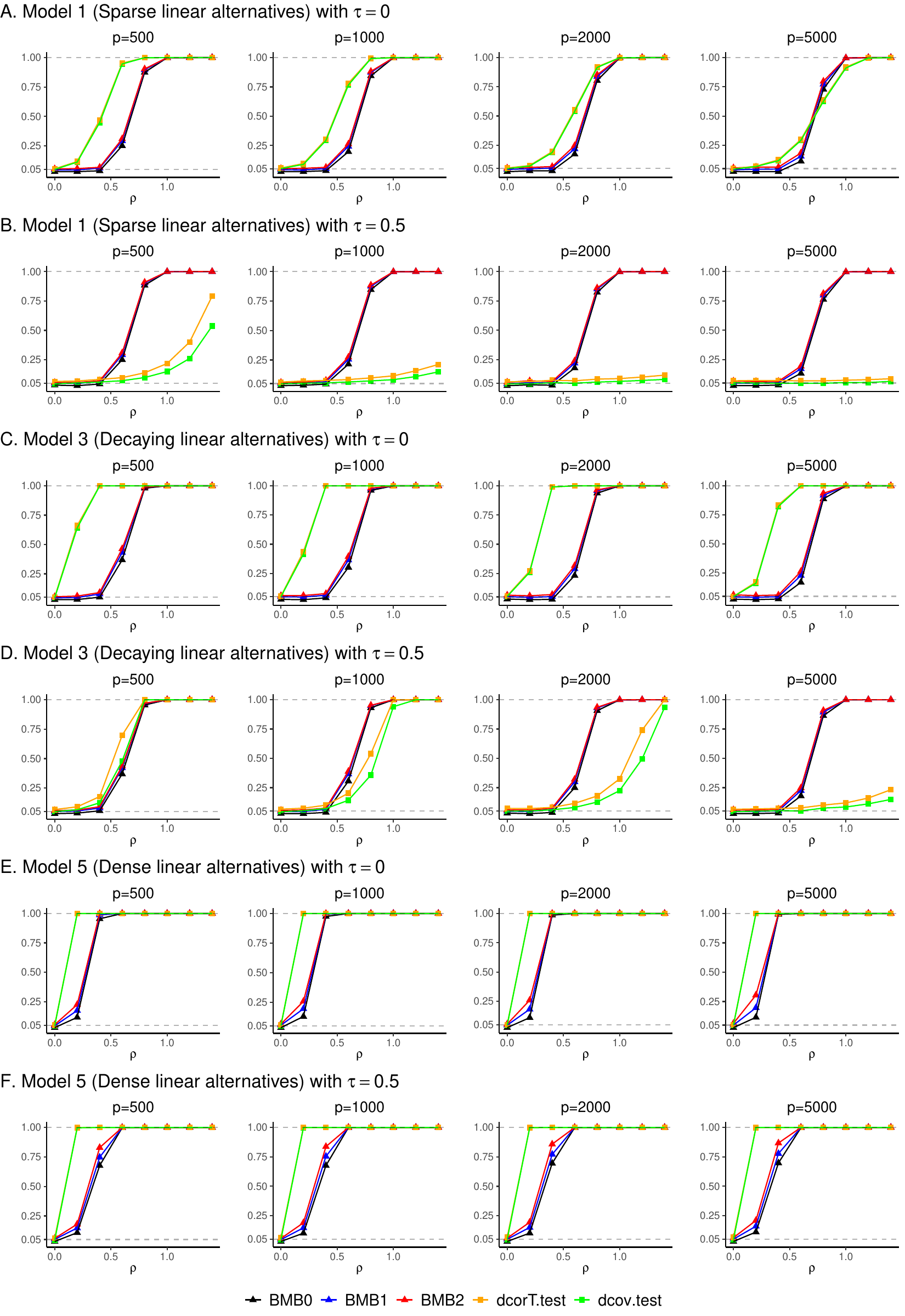}
	\caption{Experiments 2 and 4 for \textbf{linear alternatives} with $n=1000$.}
	\footnotesize{Notes: The tests have nominal level of 5\%. Results are based on $S=5,000$ Monte Carlo draws. The BMB tests and the distance covariance test used $B=499$ bootstrap samples/replicates.}
\end{figure}

\begin{figure}
	\centering
	\includegraphics[height=0.95\textheight]{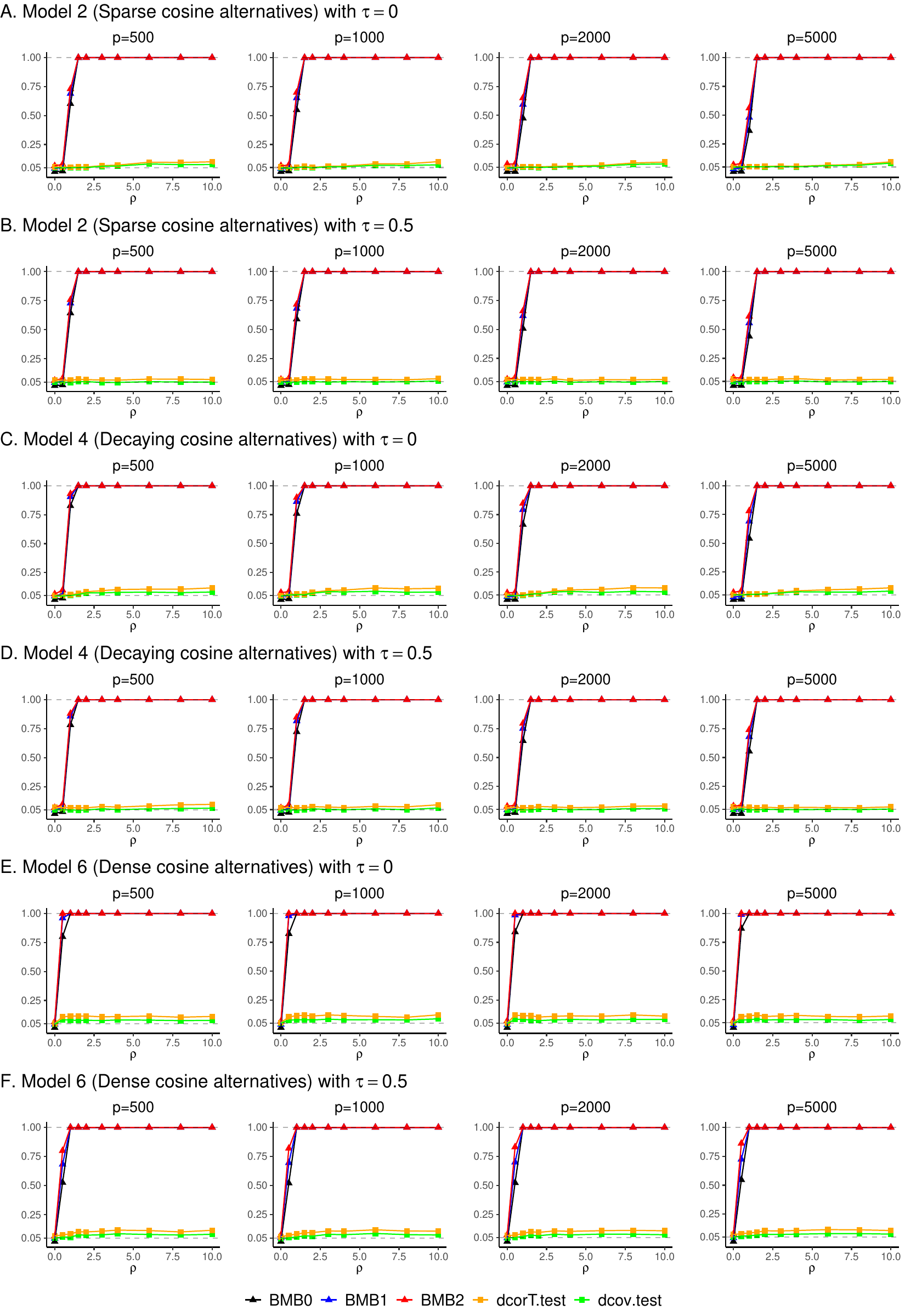}
	\caption{Experiments 2 and 4 for \textbf{cosine alternatives} with $n=200$.}
	\footnotesize{Notes: The tests have nominal level of 5\%. Results are based on $S=5,000$ Monte Carlo draws. The BMB tests and the distance covariance test used $B=499$ bootstrap samples/replicates.}
\end{figure}

\begin{figure}
	\centering
	\includegraphics[height=0.95\textheight]{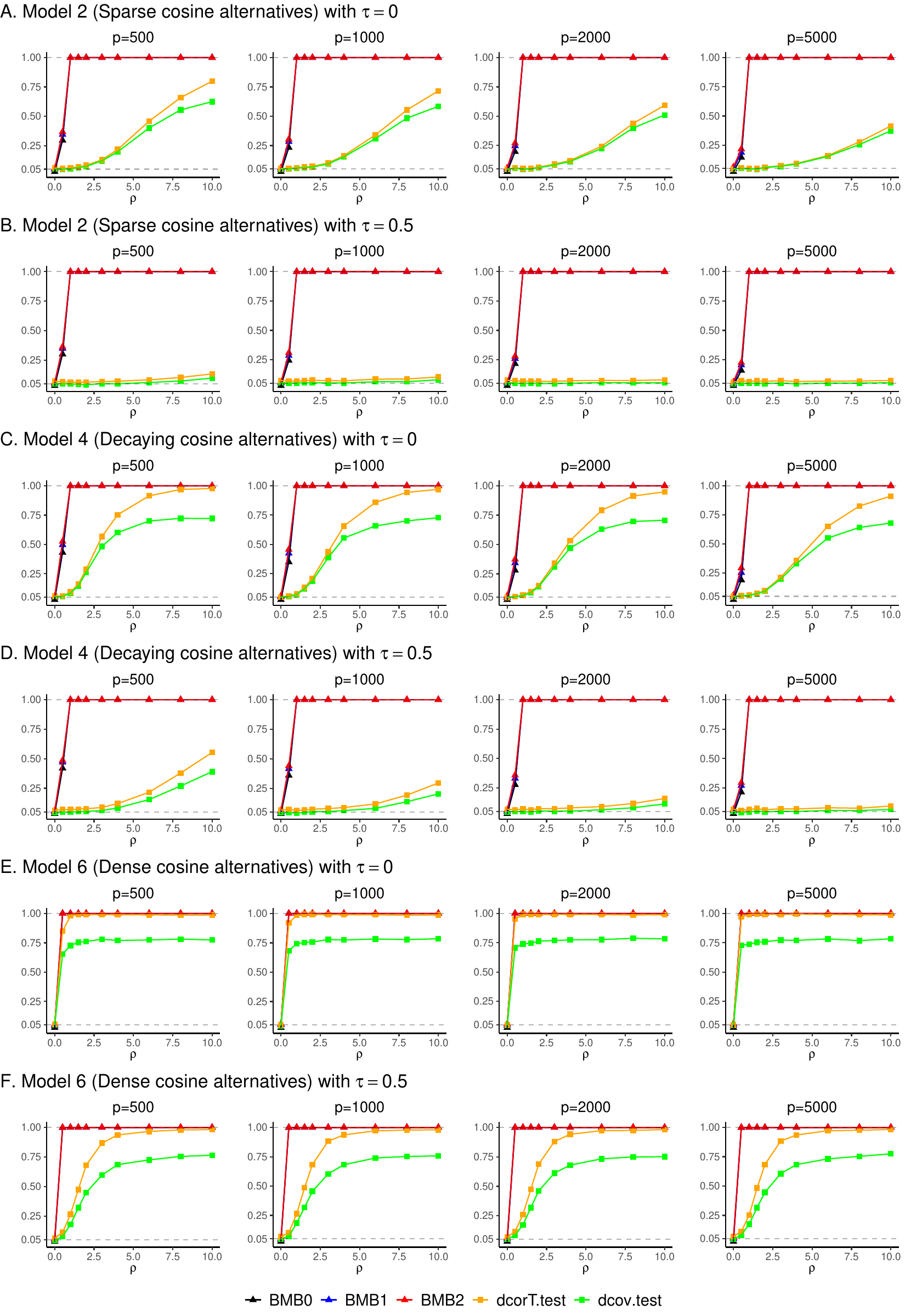}
	\caption{Experiments 2 and 4 for \textbf{cosine alternatives} with $n=1000$.}
	\footnotesize{Notes: The tests have nominal level of 5\%. Results are based on $S=5,000$ Monte Carlo draws. The BMB tests and the distance covariance test used $B=499$ bootstrap samples/replicates.}
	\label{fig:SimII-end}
\end{figure}

}{
% Do nothing for journal version (separate PDF)
}

% ------------------------------------------------- %

%                     REFERENCES                    %

% ------------------------------------------------- %

\clearpage
\bibliographystyle{economet}
\bibliography{References}

% --------------------------------------------- %

%                     EL FIN                    %

% --------------------------------------------- %

\end{document}